\documentclass[11pt]{article}
\usepackage{geometry}
\usepackage{amsmath,amsfonts,amssymb,amsthm}
\usepackage{times}
\usepackage{graphicx}
\usepackage{bbm}
\usepackage{verbatim}
\usepackage{enumitem}
\usepackage{hyperref,color}
\usepackage[capitalize,nameinlink]{cleveref}
\usepackage[dvipsnames]{xcolor}
\hypersetup{
	colorlinks=true,
	pdfpagemode=UseNone,
    citecolor=OliveGreen,
    linkcolor=blue!70!black,
    urlcolor=blue!70!black,
	pdfstartview=FitW
}
\usepackage{appendix}
\crefname{appsec}{Appendix}{Appendices}

\newtheorem{thm}{Theorem}

\newtheorem{lem}[thm]{Lemma}

\newtheorem*{thmmain}{\cref{thm:main}}
\newtheorem*{thmhighdimensional}{\cref{thm:high-dimensional}}
\newtheorem*{thmspectraleig}{\cref{thm:spectraleig}}
\newtheorem*{lemrecinf}{\cref{lem:rec-inf}}
\newtheorem*{thmlambda}{\cref{thm:lambda1}}
\newtheorem*{lemhatinf}{\cref{lem:hatinf-rec-bound}}
\newtheorem*{lemMleq}{\cref{lem:M<=I+hatI}}
\newtheorem*{lemboundratio}{\cref{lem:bound-ratio}}

\crefname{lem}{Lemma}{Lemmas}
\crefname{thm}{Theorem}{Theorems}

\theoremstyle{definition}
\newtheorem{defn}[thm]{Definition}
\newtheorem{eg}[thm]{Example}
\newtheorem{obs}[thm]{Observation}

\crefname{defn}{Definition}{Definitions}

\theoremstyle{remark}

\newcommand{\zb}{\mathbf{z}}
\newcommand{\Ab}{\mathbf{A}}
\newcommand{\Db}{\mathbf{D}}
\newcommand{\pib}{\boldsymbol{\pi}}
\newcommand{\ones}{\mathbf{1}}

\newcommand{\one}{\mathbbm{1}}

\newcommand{\TV}[2]{\big\|#1 - #2\big\|_{\mathrm{TV}}}

\newcommand{\Influence}{\mathcal{I}}
\newcommand{\I}{\Influence}

\newcommand{\Tsaw}{T_{\textsc{saw}}}
\newcommand{\Tmix}{T_{\mathrm{mix}}}

\newcommand{\eps}{\varepsilon}

\newcommand{\bs}{\backslash}

\newcommand{\N}{\mathbb{N}}

\newcommand{\R}{\mathbb{R}}

\newcommand{\II}{\mathcal{I}}
\newcommand{\JJ}{\mathcal{J}}

\newcommand{\LL}{\mathcal{L}}
\newcommand{\MM}{\mathcal{M}}
\newcommand{\PP}{\mathbb{P}}
\newcommand{\Pc}{\mathcal{P}}
\newcommand{\wPc}{\widehat{\mathcal{P}}}

\newcommand{\sra}{\text{\scriptsize{\,$\rightarrow$\,}}}
\newcommand{\seq}{\text{\scriptsize{\,$=$\,}}}
\newcommand{\sneq}{\text{\scriptsize{\,$\neq$\,}}}

\begin{document}

\title{Rapid Mixing for Colorings via Spectral Independence}
\author{
Zongchen Chen\thanks{School of Computer Science, Georgia Institute of Technology, USA. Research supported in part by NSF grant CCF-1563838.}
\and
Andreas Galanis\thanks{
  Department of Computer Science, University of Oxford, Wolfson Building, Parks Road, Oxford, OX1~3QD, UK.}
\and
Daniel \v{S}tefankovi\v{c}\thanks{Department of Computer Science, University of Rochester, USA. Research supported in part by NSF grant CCF-1563757.}
\and
Eric Vigoda$^\star$
}

\maketitle
\thispagestyle{empty}
\begin{abstract} 
The spectral independence approach of Anari et al. (2020) utilized 
recent results on high-dimensional expanders of Alev and Lau (2020) and established  rapid mixing of the 
Glauber dynamics for the hard-core model defined on weighted independent sets.
We develop the spectral independence approach for colorings, and obtain new algorithmic results for the corresponding counting/sampling problems.

Let $\alpha^*\approx 1.763$ denote the solution to $\exp(1/x)=x$ and let $\alpha>\alpha^*$.
We prove that, for any triangle-free graph $G=(V,E)$ with maximum degree $\Delta$, for all $q\geq\alpha\Delta+1$,
the mixing time of the Glauber dynamics for $q$-colorings is polynomial in $n=|V|$, with the exponent of the polynomial independent of $\Delta$ and $q$.
In comparison, previous approximate counting results for colorings
held for a similar range of $q$ (asymptotically in $\Delta$) but with
larger girth requirement or with a running time where the polynomial exponent depended on~$\Delta$ and $q$ (exponentially). 
One further feature of using the spectral independence approach to study colorings is that it avoids many of the technical complications in previous approaches caused by coupling arguments or by passing to the complex plane;  the key improvement on the running time is based on relatively simple combinatorial arguments which are then  translated into spectral bounds. 
\end{abstract}

\newpage

\clearpage
 \setcounter{page}{1}

\section{Introduction}
The colorings model is one of the most-well studied models  in computer science, combinatorics, and statistical physics. Here, we will be interested in designing efficient algorithms for sampling colorings uniformly at random.  More precisely, given a graph $G=(V,E)$ of maximum degree $\Delta$ and an integer $q\geq 3$, let $\Omega$
denote the set of proper $q$-colorings of $G$; the goal is to
generate a coloring uniformly at random (u.a.r.) from $\Omega$ in time polynomial in $n=|V|$.  The colorings model can be interpreted as a ``spin system'', when we view colors as spins with interactions between spins induced by forbidding neighboring vertices to be assigned the same spin. Note, the colorings model is a multi-spin system, in contrast to 2-spin systems such as the hard-core and the Ising models.

For spin systems, 
the key algorithmic task for studying the equilibrium properties of the model is
sampling from the associated Gibbs distribution.  For integer $q\geq 2$, the Gibbs distribution
of a $q$-spin system on an $n$-vertex graph $G$ is defined on the $q^n$ possible assignments of the spins to the vertices of the graph, where the weight of a spin assignment is determined by nearest-neighbor interactions; our goal is a sampling algorithm with running time
polynomial in $n$.  An efficient approximate sampler is polynomial-time equivalent to  an efficient approximation scheme for the corresponding  partition function~\cite{JVV,SVV, Hub, Kol}, which is the normalizing
factor in the Gibbs distribution.  

The classical approach 
for the approximate sampling/counting problem is the {\em Markov Chain Monte Carlo (MCMC)} approach, 
where we design a Markov chain whose stationary distribution 
is the Gibbs distribution.  A particularly popular Markov chain is the 
{\em Glauber dynamics}. 
Due to its simplicity and easy applicability, it is also studied as an idealized 
model for how the physical system approaches equilibrium. The Glauber dynamics updates the spin at a random vertex based on its marginal distribution in
the Gibbs distribution conditional on the spins of its neighbors. The Glauber dynamics $(X_t)$
is quite simple to describe for the colorings problem. Starting from an arbitary coloring $X_0\in \Omega$, at time $t\geq 0$,    
choose a vertex $v$ u.a.r. and then set $X_{t+1}(w)=X_t(w)$ for all $w\neq v$ and choose
$X_{t+1}(v)$ u.a.r. from the set of colors that do not appear in the neighborhood of $v$. The key quantity for the Glauber dynamics is the 
{\em mixing time} which is the number of steps from the worst initial state $X_0$ to reach within 
total variation distance $\leq 1/4$ of its stationary distribution. Despite its simplicity, analyzing the mixing time of the Glauber dynamics even 
for the canonical case of the colorings model is surprisingly challenging.

There are two non-MCMC algorithmic methods that have been 
powerful and more amenable to a finer understanding so far: the correlation decay and Barvinok's interpolation methods.
The basis of the correlation decay method is the so-called {\em strong spatial mixing (SSM)} condition\footnote{Roughly speaking, the SSM condition captures whether, if we fix two partial assignments $\sigma,\tau$ on a subset of vertices $T$, the difference in the conditional marginal distribution at a vertex $v$ decays exponentially in the distance between $S$ and $v$, where $S\subseteq T$ is the subset of vertices that $\sigma,\tau$ differ.}; for 2-spin systems, one for example can utilize SSM together with a clever tree construction of Weitz~\cite{Weitz}
to efficiently estimate marginals and hence obtain an approximation algorithm.  The alternative algorithmic method by Barvinok~\cite{Barvinok}, which was further refined by Patel and Regts~\cite{PR17},   examines instead the roots
of the partition function in the complex plane and approximates the Taylor series of
the partition function in a zero-free~region. 

Both of these non-MCMC approaches have been shown 
to work for  antiferromagnetic 2-spin systems\footnote{A 2-spin system is called antiferromagnetic if neighboring spins prefer to be different, see for example~\cite{LLY13} for more details. Examples include the hard-core model and the antiferromagnetic Ising model.} up to the so-called tree uniqueness threshold, see \cite{Weitz,SST,LLY13} for the correlation decay approach and \cite{PR,SS} for the  interpolation method; see also~\cite{Sly,SlySun,GSV:ising} for complementary hardness results. However, the running time of these algorithmic approaches scales as $O(n^C)$ where the exponent $C$ depends on $\Delta$ and on the multiplicative gap $\delta$ from the tree uniqueness threshold; obtaining faster algorithms even for 2-spin systems is a major open problem. 

To this vein, MCMC methods typically give much faster (randomized) algorithms, however corresponding results were lacking until a recent breakthrough result of Anari, Liu and Oveis Gharan~\cite{ALO}, who proved rapid mixing of the Glauber dynamics for the hard-core model, matching the parameter range of the aforementioned non-MCMC approaches and also improving the running time with a polynomial exponent which is independent of the degree bound $\Delta$. They introduced a spectral
independence approach which utilizes high-dimensional expander results of Alev and Lau~\cite{AL} (cf. \cite{KO, Oppenheim1}).
The work of~\cite{ALO} establishes that, for 2-spin systems, it suffices to bound 
the largest eigenvalue of the $n\times n$ influence matrix $\I$ where the $(v,w)$ entry captures the influence of the
fixed spin at vertex $v$ on the marginal probability at vertex $w$; 
we explain this in more detail in \cref{sec:rough-outline}. The running time of the result of \cite{ALO} was further improved in~\cite{CLV}, who also generalised the approach to antiferromagnetic 2-spin systems up to the tree-uniqueness threshold by  showing how to utilize potential-function arguments that were previously used to establish SSM.

Going beyond 2-spin systems, all of these methods become harder to control even well above the tree-uniqueness threshold, $q=\Delta+1$, which marks the onset of computational hardness (even for triangle-free graphs, see~\cite{GSV}). Let $\alpha^*\approx 1.763$ be the solution to $\exp(1/x)=x$; this threshold has appeared in several related results for colorings, though obtaining corresponding algorithms has been challenging. For example, for $\alpha>\alpha^*$,  Gamarnik, Katz, and Misra~\cite{GKM} proved SSM on triangle-free
graphs when $q>\alpha\Delta+\beta$ for some constant $\beta=\beta(\alpha)$; see also \cite{GMP} for a related result on amenable graphs. However, the correlation decay approach has so far yielded an efficient algorithm only for $q\geq 2.58\Delta$, see~\cite{GK,LY}. It was not until recently that the SSM result of \cite{GKM} was converted to an algorithm for triangle-free graphs by Liu, Sinclair, and Srivastava~\cite{LSS} utilizing the complex
zeros approach; however, just as for 2-spin systems,  the polynomial exponent in the running time depends exponentially on $\Delta$ and the distance of $\alpha$ from~$\alpha^*$. 

The analysis of Glauber dynamics for colorings has not been easier. Jerrum~\cite{Jerrum} proved that the mixing time is $O(n\log{n})$ for all graphs when $q>2\Delta$. 
This was improved to $q>\frac{11}{6}\Delta$ with mixing time $O(n^2)$ by Vigoda~\cite{Vigoda}, which was only recently improved to $q>(\frac{11}{6}-\delta)\Delta$ for a small constant $\delta>0$~\cite{CDMPP}. Back to asymptotic results, for $\alpha>\alpha^*$ and large degrees $\Delta>\Delta_0(\alpha)$,  Dyer~et al.~\cite{DFHV} showed that on graphs with girth $\geq 5$ and maximum degree $\Delta$ the mixing time of the Glauber dynamics is $O(n\log{n})$ using sophisticated coupling arguments building upon local uniformity results of Hayes~\cite{Hayes}. See \cite{DFHV,HV} for improvements by imposing other degree/girth restrictions.

Our main contribution is to develop the spectral independence approach of \cite{AL,ALO} for colorings, and analyze Glauber dynamics in the regime $q\geq \alpha \Delta+1$ for all $\alpha>\alpha^*$ on triangle-free graphs. Our result applies for all $\Delta$ and we show that the exponent of the mixing time does not depend on $\Delta$ and $q$, 
yielding  substantially faster randomized algorithms for sampling/counting colorings than the previous deterministic ones (at the expense of using randomness).  
\newcommand{\statethmmain}{Let $\alpha^*\approx 1.763$ denote the solution to $\exp(1/x)=x$.  For all $\alpha>\alpha^*$, there exists $c=c(\alpha)>0$ such that, for any triangle-free graph $G=(V,E)$ with maximum degree $\Delta$ and any integer $q\geq \alpha \Delta+1$, the mixing time of the
Glauber dynamics on $G$ with $q$ colors is at most $n^{c}$, where $n=|V|$.}
\begin{thm}
\label{thm:main}
\statethmmain
\end{thm}
One feature of using the spectral independence approach to study colorings is that it avoids many of the technical complications caused by coupling arguments or by passing to the complex plane, and allows us to get a better grip on the quantities of interest (marginals); indeed, as we shall explain in the next section, the key improvement on the running time is inspired by relatively simple combinatorial arguments and translating them into appropriate spectral bounds.

\subsection{Proof approach}
\label{sec:rough-outline}

Our work builds upon the spectral independence approach introduced by 
Anari, Liu, and Oveis Gharan~\cite{ALO}, which in turn utilizes the high-dimensional expander
work of Alev and Lau~\cite{AL}.  Consider a graph $G=(V,E)$ of maximum degree $\Delta$.
The key to this approach is to analyze the spectral radius of the $nq\times nq$ matrix 
$\MM$ where, for distinct $v,w\in V$ and $i,k\in [q]$, 
\[ \MM\big((v,i),(w,k)\big) = 
\PP(\sigma_w \seq k \mid \sigma_v \seq i) - \PP(\sigma_w \seq k).
\]
The spectral independence approach is formally presented in \cref{sec:outline},
and the connection to rapid mixing is formally stated in \cref{thm:high-dimensional}.

To be precise, in the spectral independence approach
we need to analyze the corresponding matrix $\MM$ for the Gibbs distribution $\mu_G$ conditional on 
all fixed assignments $\sigma_S$ for all $S\subseteq V$.  
A fixed assignment $\sigma_S$ yields a list-coloring problem instance and hence we need to consider the
more general list-coloring problem.  
At a high-level this is analogous to SSM (strong spatial mixing). 
We do not formally define SSM in this paper since it is not explicitly used.   Roughly speaking, in SSM we
consider the effect 
of a pair of boundary colorings on the marginal distribution of a specified vertex given the worst fixed assignment $\sigma_S$ 
for an arbitrary subset $S$.

In~\cite{CLV} it was shown for $2$-spin systems how the standard proof approach for establishing SSM also implies spectral independence. 
However, the restriction to $2$-spin systems is fundamental.
For $2$-spin systems, Weitz~\cite{Weitz} showed that for any graph $G=(V,E)$, any $v\in V$, there
is an appropriately defined tree $T=\Tsaw(G,v)$ (corresponding to the self-avoiding walks in $G$ starting from $v$ with a 
 particularly fixed assignment to the leaves) so that the marginal distribution for the root of $T$ (in the corresponding Gibbs distribution $\mu_T$)
is identical to the marginal distribution for $v$ (in $\mu_G$).  Utilizing this self-avoiding walk tree construction, the main idea in proofs establishing SSM is
to design a potential function on the ratio of the marginal distribution for the root of a tree and prove that this potential function is contracting 
for the corresponding tree recursions.

Gamarnik, Katz, and Misra~\cite{GKM} established SSM for the colorings problem when $k>\alpha^*\Delta+\beta_1$ for some constant $\beta_1>0$
for all triangle-free graphs of maximum degree $\Delta$.  Even though Weitz's self-avoiding walk tree connection no longer holds for colorings,~\cite{GKM} 
utilized an appropriately constructed computation tree for the more general list-coloring problem.  They then present a potential function
which is contracting with respect to the corresponding recursions for their computation tree.  

Previous proofs for the spectral independence study entries of the influence matrix using the derivative of the potential function. 
Instead, the SSM proof approach of \cite{GKM} uses a non-differentiable potential function so we cannot use the same analytical approach. 
We analyze the entries of the influence matrix by a more combinatorial argument, paying attention to the entries that are potentially large and therefore corresponds to highly correlated vertex-spin pairs. 

In particular, to bound the spectral radius of the matrix~$\MM$,
we consider the following quantity: for a pair of vertices $v,w\in V$ and a color $k\in [q]$, define the {\em maximum influence of $v$ on $(w,k)$} as:
\[
\II[v \sra (w,k)] 
= \max_{i,j\in [q]} 
\left| \PP (\sigma_w \seq k \mid \sigma_v \seq i) 
- 
\PP(\sigma_w \seq k \mid \sigma_v \seq j) \right|. 
\]
This is reminiscent of the potential function given in \cite{GKM} and an adaptation of their arguments allows us to write a recursion for $\II[v \sra (w,k)]$, expressing it in terms of the influences of the neighbors of $v$ in a graph where $v$ is deleted. In turn, this gives a recursion for the aggregate influences (over $w,k$); the growth rate of the aggregate influences in the recursion is controlled by the product of the degree of $v$ and the marginal probability at $v$ and the condition $q\geq \alpha \Delta+1$ guarantees that this product is less than 1. The end result of this ``vanilla'' approach yields that the spectral radius of $\MM$ is  $C\Delta/\eps$ when $q\ge(1+\eps)\alpha^*\Delta+1$ for arbitrarily small $\eps>0$ and $C$ is an absolute constant. This in turn gives a (weaker) polynomial bound for fixed values of $\Delta$ (the constant in the exponent grows linearly with $\Delta$). 
While this argument does not quite give what we want, it contains many of the relevant ideas that are used in the more refined argument later, so we present the simpler argument in \cref{sec:simpler-bound}.

To get the stronger polynomial bound stated in \cref{thm:main} for all $\Delta$, we need instead to prove that the spectral radius of $\MM$ is independent of $\Delta$ and $q$; 
achieving this stronger result requires further insight.  For the influences $\MM$ the only large entries are the ``diagonal'' entries 
corresponding to the cases when $i=k$.  This is illustrated by the simple example of a star on $\Delta+1$ vertices in \cref{sec:nottight}
where these diagonal entries are of order $\Theta(1/q)$ whereas the non-diagonal entries are $O(1/q^2)$.
To handle this discrepancy 
we introduce a new notion of maximum influence $\hat{\II}_L[v \sra (w,k)]$ corresponding to the cases $i,j\neq k$. 
We need a more intricate induction argument to simultaneously maintain appropriate bounds on both of these two quantities.
The final result upper bounds the row-sum of $\MM$ by $O((\Delta/q) \eps^{-2})$.
This proof which is the main ingredient of the proof of \cref{thm:main} is presented in \cref{sec:better-bound}.

\section{Spectral independence and proof outline}
\label{sec:outline}

\subsection{Preliminaries}
Let $q\geq 3$ be an integer and denote by $[q]:=\{1,\hdots,q\}$. 

A list-coloring instance is a pair $(G,L)$ where $G=(V,E)$ is a graph and  $L=\{L(v)\}_{v\in V}$ prescribes a list $L(v)\subseteq[q]$ of available colors for each $v\in V$; it will also be convenient to assume that the vertices of $G$ are ordered by some relation $<$ (the ordering itself does not matter). 
A proper list-coloring for the instance $(G,L)$ is an assignment $\sigma:V\rightarrow [q]$ such that 
$\sigma_v \in L(v)$ for each $v\in V$ and $\sigma_v \neq \sigma_w$ for each $\{v,w\} \in E$. 
The instance is satisfiable iff such a proper list-coloring exists. Note, $q$-colorings corresponds to the special case where $L(v)=[q]$ for each $v\in V$.  
For a satisfiable list-coloring instance $(G,L)$, we will denote by $U_{G,L}$ the set $\{(v,i)\mid v\in V, i\in L(v)\}$, by $\Omega_{G,L}$ the set of all proper list-colorings, and by $\PP_{G,L}$ the uniform distribution over $\Omega_{G,L}$; we will omit $G$ from notations when it is clear from context. We typically use $\sigma$ to denote a random list-coloring that is distributed according to $\PP_{G,L}$.

We will be interested in analyzing the Glauber dynamics on  $\Omega_{G,L}$. This is a Markov chain $(Z_t)_{t\geq 0}$ of list-colorings which starts from an arbitrary $Z_0\in \Omega_{G,L}$ and at each time $t\geq 0$ updates the current list-coloring $Z_{t}$ to $Z_{t+1}$ by selecting a vertex $v\in V$ u.a.r. and setting $Z_{t+1}(v)=c$, where $c$ is a color chosen u.a.r. from the set $L(v)\backslash Z_t(N_G(v))$; for a vertex $w\neq v$, the color of $w$  is unchanged, i.e.,  $Z_{t+1}(w)=Z_{t}(w)$. The transition matrix of the Glauber dynamics will be denoted by $\Pc=\Pc_{G,L}$.

To ensure satisfiability of $(G,L)$ as well as ergodicity of the Glauber dynamics, we will henceforth assume the well-known condition that $|L(v)|\geq \Delta_G(v)+2$ for all $v\in V$, where $\Delta_G(v)=|N_G(v)|$ and $N_G(v)$  is the set of  neighbors of $v$ in $G$.\footnote{\label{foot:ergod}To ensure satisfiability, it suffices to have the assumption $|L(v)|\geq \Delta_G(v)+1$ for all $v\in V$; in fact, for every $v\in V$ and $i\in L(v)$ there exists a list-coloring $\sigma$ of $(G,L)$ with $\sigma_v=i$. The slightly stronger condition $|L(v)|\geq \Delta_G(v)+2$ for every $v\in V$ ensures that any two list-colorings $\sigma,\tau$ are ``connected'' by a sequence of list-colorings where consecutive list-colorings differ at the color of a single vertex. (A clique with $q+1$ vertices gives a counterexample to this latter property for $q$-colorings).} Then, Glauber dynamics converges to the uniform distribution over $\Omega_{G,L}$. The mixing time of the chain is the number of steps needed to get within total variation distance $\leq 1/4$ from a worst-case initial state, i.e.,
\[\Tmix=\max_{\sigma\in \Omega_{G,L}} \min \Big\{t \ge 0\,\Big|\, X_0=\sigma, \TV{X_t}{\PP_{G,L}}\leq 1/4\Big\}.\]
It is well-known that, for any integer $k\geq 1$, after $k\Tmix$ steps the total variation distance from the stationary distribution is no more than $(1/2)^{k+1}$; see, e.g., \cite[Chapter 4]{LP}. Let $\lambda_2(\Pc)$ be the second largest eigenvalue\footnote{More generally, for a square matrix $M\in \mathbb{R}^{n\times n}$ all of whose eigenvalues are real, we let $\lambda_1(M),\lambda_2(M),\hdots,\lambda_n(M)$ denote the eigenvalues of $M$ in non-increasing order.} of $\Pc$, and since the Glauber dynamics on $(G,L)$ is reversible, irreducible, and aperiodic, we have the following bound by applying well-known results from the theory of Markov chains. 
\begin{lem}[see, e.g., {\cite[Theorem 12.3 \& 12.4]{LP}}]
Let $(G,L)$ be a list-coloring instance with $G=(V,E)$ and $L=\{L(v)\}_{v\in V}$. 
Let $n=|V|$ and $Q=\max_{v\in V}|L(v)|$. 

Then, denoting by $\lambda_2=\lambda_2(\Pc_{G,L})$ the second largest eigenvalue of $\Pc_{G,L}$, we have that the mixing time of the Glauber dynamics satisfies $\Tmix\leq \frac{n\ln (4Q)}{1-\lambda_2}$.
\end{lem}


\subsection{Local expansion for list-colorings and connection to Glauber dynamics}
To analyze the Glauber dynamics on a list-coloring instance $(G,L)$, we will use the spectral independence approach of \cite{AL,ALO}. The key ingredient in this approach is to give a bound on the spectral gap of a random walk on an appropriate weighted graph; here we explain how these pieces can be adapted in the list-coloring setting and state the main result that allows us to conclude fast mixing of Glauber dynamics.

\begin{defn} Let $H_{G,L}$ be the weighted graph with vertex set $U_{G,L}$ and edges $\{(v,i),(w,k)\}$ for all $(v,i),(w,k)\in U_{G,L}$ with $v\neq w$, with corresponding edge weight  $\PP_{G,L} (\sigma_v \seq i,\sigma_w \seq k)$. 

Let $\widehat{\Pc}_{G,L}$ be the transition matrix of the simple non-lazy random walk on $H_{G,L}$.
\end{defn}
\begin{defn}
For $\alpha\in [0,1]$, we say that $(G,L)$ has local expansion bounded by  $\alpha$ if the second largest eigenvalue of the simple non-lazy random walk on the weighted graph $H_{G,L}$ is at most $\alpha$, i.e.,  $\lambda_{2}\big(\wPc\big)\leq \alpha$ where $\wPc=\widehat{\Pc}_{G,L}$ is the transition matrix of the random walk. 
\end{defn}
For the spectral independence approach of~\cite{AL,ALO}, we will need to consider conditional distributions of $\PP_{G,L}$ given a partial list-coloring\footnote{\label{foot:partial}For a subset $S\subseteq V$, we say that $\tau$ is a partial list-coloring of $(G,L)$ on $S$ if $\tau=\sigma_S$ for some $\sigma\in \Omega_{G,L}$.} on a subset of vertices; this setting is reminiscent of SSM, though the goal is different.  For a partial list-coloring $\tau$ on a subset $S\subseteq V$, let $(G_{\tau}, L_{\tau})$ be the list-coloring instance on the induced subgraph $G[V\backslash S]$ with lists obtained from $L$ by removing the unavailable colors that have been assigned by $\tau$ for each vertex in $V\backslash S$, i.e.,  $L_{\tau}=\{L_{\tau}(v)\}_{v\in V\backslash S}$ where for $v\in V\backslash S$ we have  $L_{\tau}(v) = L(v)\backslash \tau(N_G(v)\cap S)$.

To capture those instances of list-colorings obtained from an instance of $q$-colorings by assigning fixed colors to a subset of vertices, the following  notion of $(\Delta,q)$-list-colorings will be useful.
\begin{defn}
Let $\Delta,q$ be positive integers with $\Delta \ge 3$ and $q\geq \Delta+2$. We say that  $(G,L)$  is a $(\Delta,q)$-list-coloring instance if $G=(V,E)$ has maximum degree $\Delta$ and for each $v\in V$  it holds that  $L(v)\subseteq [q]$ and $|L(v)| \ge q-\Delta+\Delta_G(v)$.
\end{defn}

We are now ready to state the spectral independence approach for list-colorings.

\newcommand{\statethmhighdimensional}{Let $(G,L)$ be a $(\Delta,q)$-list-coloring instance where $G$ is an $n$-vertex graph. Suppose that for each integer $s=0,1,\dots,n-2$ there is $\ell_s\in [0,1)$ such that for every partial list-coloring $\tau$ on a subset $S\subseteq V$ with $|S|=s$, the conditioned instance $(G_\tau,L_\tau)$ has local expansion bounded by $\ell_s$.

Then, for $L:=\prod^{n-2}_{s=0} (1-\ell_s)^{-1}$, the spectral gap of the Glauber dynamics on $(G,L)$ is at least $1/(nL)$ and its mixing time is at most $Ln^2\ln(4q)$.}
\begin{thm}\label{thm:high-dimensional}
\statethmhighdimensional
\end{thm}
\subsection{Key lemmas: establishing local expansion for list-colorings}
The hard part for us is to verify the conditions of \cref{thm:high-dimensional}, i.e., bound the local expansion of a (conditioned) list-coloring instance. To do this the following matrix will help us to concentrate on the non-trivial eigenvalues of the corresponding random walk.
\begin{defn}\label{def:matrixM}
Let  $(G,L)$ be a list-coloring instance. Let $\MM=\MM_{G,L}$ be the square matrix with indices from the set $U_{G,L}$, where the entry indexed by $(v,i),(w,k)\in U_{G,L}$ is $0$ if $v=w$, and 
\begin{equation*}
\MM\big((v,i),(w,k)\big) = 
\PP_{G,L} (\sigma_w \seq k \mid \sigma_v \seq i) - \PP_{G,L} (\sigma_w \seq k), \quad \text{if $v \neq w$}.
\end{equation*}
\end{defn}

\newcommand{\statethmspectraleig}{Let $(G,L)$ be a list-coloring instance with $G=(V,E)$ and $L=\{L(v)\}_{v\in V}$ such that $|L(v)|\geq \Delta_G(v)+2$ for all $v\in V$, and $n=|V|\geq 2$. Let $\wPc$ be the transition matrix of the simple non-lazy random walk on the weighted graph $H_{G,L}$. Then, the eigenvalues of $\MM$ are all real and $\lambda_2(\wPc)=\frac{1}{n-1}\lambda_1(\MM)$ where $\MM=\MM_{G,L}$ is the matrix from \cref{def:matrixM}.}

In \cref{sec:spectraleig}, we show that the second largest eigenvalue of $\wPc$ can be studied by focusing on the largest eigenvalue of $\MM$.
\begin{thm}\label{thm:spectraleig}
\statethmspectraleig
\end{thm}
\cref{thm:spectraleig} follows from spectral arguments and is inspired from ideas about $d$-partite simplicial complexes  in \cite{ALO,Oppenheim1}. Then, the  core of our argument behind the proof of \cref{thm:main} is to establish the following bound on $\lambda_1(\MM)$ by studying the list-coloring distribution.

\newcommand{\statethmlambda}{Let $\eps>0$ be arbitrary, and suppose that  $(G,L)$ is a $(\Delta,q)$-list-coloring instance with $q\geq (1+\eps)\alpha^*\Delta+1$ and $G$ a triangle-free graph. Then, $\lambda_1(\MM)\leq 64 \big( \frac{1}{\eps}+1 \big)^2 \frac{\Delta}{q}$  where $\MM=\MM_{G,L}$ is the matrix from \cref{def:matrixM}.}
\begin{thm}\label{thm:lambda1}
\statethmlambda
\end{thm}

\subsection{Combining the pieces: proof of \texorpdfstring{\cref{thm:main}}{Theorem 1}}
\label{sec:proofmain}
Assuming \cref{thm:high-dimensional,thm:spectraleig,thm:lambda1}, we can complete here the proof of~\cref{thm:main}.
\begin{thmmain}
\statethmmain
\end{thmmain}
\begin{proof}
We may assume that $\alpha < 2$, otherwise the result follows from \cite{Jerrum}. Let $\eps>0$ be such that $\alpha = (1+\eps)\alpha^*$. We will show the result with $c=80C_\alpha^2$ where $C_{\alpha}=\frac{64}{\alpha} \big( \frac{1}{\eps}+1 \big)^2$. 
Suppose that $G$ is an $n$-vertex triangle-free graph with maximum degree $\Delta$, and $q\geq \alpha \Delta+1$. Again, from the result of \cite{Jerrum} we may assume that $q \le 2 \Delta$.   If $n=1$ the result is immediate, so assume $n\geq 2$ in what follows. Let $C=64 \big( \frac{1}{\eps}+1 \big)^2 \frac{\Delta}{q}$ be the bound from \cref{thm:lambda1}, and note that $1<C\le C_{\alpha}$.

Consider the list-coloring instance $(G,L)$ where $L(v)=[q]$ for each $v\in V$. Then, Glauber dynamics with $q$ colors on $G$ is the same as Glauber dynamics on $(G,L)$, so it suffices to bound the mixing time of the latter.  We will show that \cref{thm:high-dimensional} applies with $\ell_s=\min\{\frac{C}{n-1-s},1-2(1/q^4)^{n-s}\}$ for each $s\in \{0,1,\hdots,n-2\}$. Indeed, let $\tau$ be an arbitrary partial list-coloring on $S\subseteq V$ with $|S|=s$ for some $s\in \{0,1,\hdots,n-2\}$ and consider the conditioned instance $(G_\tau,L_\tau)$ with $G_\tau=(V_\tau,E_\tau)$. Then, for every vertex $v\in V_{\tau}$ we have that $|L_\tau(v)|\geq q-\Delta+\Delta_{G_\tau}(v)$ since the conditioning on $\tau$ disallows at most $\Delta-\Delta_{G_\tau}(v)$ colors from $v$, and hence $(G_\tau,L_\tau)$ is a $(\Delta,q)$-list-coloring instance. Therefore, by~\cref{thm:spectraleig} and~\cref{thm:lambda1} applied to $(G_\tau,L_\tau)$, we obtain that $(G_\tau,L_\tau)$ has local expansion bounded by $\frac{C}{n-1-s}$. The local expansion is also bounded by $1-2(1/q^4)^{n-s}$ using  conductance arguments.\footnote{For any reversible Markov chain with transition matrix $\Pc$, it holds that $1-\lambda_2\geq \Phi^{2}/2$, where $\Phi$ is the conductance of the chain, see, e.g., \cite[Theorem 13.14]{LP}. In the proof of~\cref{thm:lambda1}, it is shown that the stationary distribution of the random walk on $H_\tau$ is given by $\{\frac{1}{n-s-1}\PP_{G_\tau,L_\tau}(\sigma_v \seq k)\}_{(v,k)\in U_{G_\tau,L_\tau}}$,  whose entries are crudely lower-bounded by $1/q^{2(n-s)}$, see~\cref{foot:ergod}. This in turn yields the desired bound on the local expansion of $(G_\tau,L_\tau)$.}  This verifies the assumptions of \cref{thm:high-dimensional}, so it follows that the mixing time of the Glauber dynamics on $G$ is at most $Ln^2\ln(4q)$, where 
$L=\prod^{n-2}_{s=0} (1-\ell_s)^{-1}$. Let $k_0=\left\lceil 2C\right\rceil\leq 3C_\alpha$, then we have that 
\[
L \le 
\left( \frac{q^{4k_0}}{2} \right)^{k_0-1} \cdot  
\prod^{n-1-k_0}_{s=0} \left(1-\frac{C}{n-s-1}\right)^{-1} 
\le q^{4k_0^2} \cdot n^{2C} 
\le n^{74 C_\alpha^2},
\] 
since $-\sum^{n-1}_{i=k_0} \ln(1-\frac{C}{i})\leq 2C\sum^{n-1}_{i=k_0} \frac{1}{i}\leq 2C \ln n$ and $q \le 2\Delta \le n^2$.  

Using the bound on $L$, \cref{thm:high-dimensional} yields that $\Tmix\leq n^{c}$ with $c=80C_\alpha^2$, finishing the proof.
\end{proof}

\noindent \textbf{Organisation of the rest of the paper.} \cref{sec:slower,sec:better-bound} are devoted to the proof of the key \cref{thm:lambda1}, and~\cref{sec:recboundmarg} finishes off a couple of left-over technical lemmas used in the proof. In \cref{sec:aux}, we give the details of the spectral independence approach for colorings and prove \cref{thm:high-dimensional,thm:spectraleig}. 

In our proofs henceforth, it will be convenient to define the following slightly more accurate  form of the region of $(\Delta,q)$ where our results apply to. 
\begin{defn}[Parameter Region $\Lambda_\eps$]
Let $\alpha^* \approx 1.763$ denote the solution to $\exp(1/x)=x$. 
For $\eps > 0$, define $\Lambda_\eps = \left\{ (\Delta,q) \in \N^2 \mid \Delta \ge 3, \; q \ge \alpha\Delta+\beta \right\}$ where $\alpha = (1+\eps) \alpha^*$ and $\beta = 2-\alpha + \frac{\alpha}{2(\alpha^2-1)} < 0.655$. 
\end{defn}

\section{Simpler proof of a slower mixing result}
\label{sec:slower}
Let $(G,L)$ be a $(\Delta,q)$-list-coloring instance as in \cref{thm:lambda1}, our goal is to bound the spectral radius of the matrix $\MM_{G,L}$ from \cref{def:matrixM}. In this section, we will prove a weaker result than the one in \cref{thm:lambda1} which already contains some of the key ideas and will motivate our refinement in \cref{sec:better-bound}. 

In particular, we will show that for $\alpha>\alpha^*$ there exists a constant $C=C(\alpha)$ such that whenever $q\geq \alpha\Delta+1$ it holds that $\lambda_1(\MM_{G,L})\leq C\Delta$. Note the dependence on $\Delta$ of this bound, in contrast to that of \cref{thm:lambda1}; mimicking the proof of \cref{thm:main} given earlier would give a mixing time bound of $O(n^{C'\Delta})$ for the Glauber dynamics for some constant $C'=C'(\alpha)>0$, which is much weaker than what \cref{thm:main} asserts. Nevertheless, we will introduce several of the relevant quantities/lemmas that will also be relevant in the more involved argument of \cref{sec:better-bound}.

 It is well-known that, for any square matrix the spectral radius is bounded by the maximum of the $L_1$-norms of the rows. In our setting, the (weaker) bound on $\lambda_1(\MM_{G,L})$ will therefore be obtained by showing that, for an arbitrary vertex $v$ of $G$ and a  color $i\in L(v)$, it holds that
 \footnote{\label{foot:ext}
 Henceforth, it will be convenient to extend  $\MM_{G,L}$ by setting $\MM_{G,L}((v,i),(w,k))=0$ when $k \notin L(w)$ or $i \notin L(v)$.}
\begin{equation}\label{eq:suma}
\sum_{w \in V \bs \{v\}} \sum_{k\in [q]} \big| \MM_{G,L} \big((v,i),(w,k)\big) \big| \leq 4\left( \frac{1}{\eps}+1 \right)\Delta.
\end{equation}
To bound the sum in \eqref{eq:suma}, we introduce the \emph{maximum influence}, which describes the maximum difference of the marginal probability of $\sigma_w = k$ under all color choices of~$v$. 

\begin{defn}[Maximum Influences]
Let $(G,L)$ be a $(\Delta,q)$-list-coloring instance. Let $v,w$ be two vertices of $G$,  and $k\in [q]$. The maximum influence of $v$ on $(w,k)$ is defined to be 
\begin{equation*}
\II_{G,L}[v \sra (w,k)] 
= \max_{i,j\in L(v)} 
\big| \PP_{G,L} (\sigma_w \seq k \mid \sigma_v \seq i) 
- 
\PP_{G,L} (\sigma_w \seq k \mid \sigma_v \seq j) \big|. 
\end{equation*}
\end{defn}

\begin{obs}\label{obs:bound}
$|\MM_{G,L}((v,i),(w,k))| \le \II_{G,L}[v \sra (w,k)]$ for all distinct $v,w \in V$, $i\in L(v)$, and $k\in [q]$.
\end{obs}
\begin{proof}
If $k\notin L(w)$, then $\MM_{G,L}((v,i),(w,k))=\II_{G,L}[v \sra (w,k)]=0$. For $k\in L(w)$, since $v\neq w$, we have $\MM_{G,L}((v,i),(w,k))=\PP(\sigma_w \seq k \mid \sigma_v \seq i) - \PP(\sigma_w \seq k)$ and so the law of total probability gives
\[\MM_{G,L}((v,i),(w,k))=\sum_{j\in L(v)}\big(\PP(\sigma_w \seq k \mid \sigma_v \seq i)-\PP(\sigma_w \seq k\mid \sigma_v \seq j)\big)\PP(\sigma_v \seq j),\]
from where the desired inequality follows.
\end{proof}

 Hence, to bound the sum in \eqref{eq:suma}, it suffices to bound the sum $\sum_{w \in V \bs \{v\}} \sum_{k\in [q]}\II_{G,L}[v \sra (w,k)]$ instead. Our ultimate goal is to write a recursion for this latter sum, bounding by an analogous sum for the neighbors of $v$ (in the graph where $v$ is deleted). To get on the right track, we start by writing a recursion for influences.

\subsection{A recursive approach to bound influences}
\label{sec:recursion}

In this section, we derive a recursion on influences. Recall that a list-coloring instance is a pair $(G,L)$ where $G=(V,E)$ is a graph, $L=\{L(v)\}_{v\in V}$ prescribes a list $L(v)$ of available colors for each $v\in V$, and the vertices of $G$ are ordered by some relation $<$.

\begin{defn}
Let $(G,L)$ be a list-coloring instance with $G=(V,E)$ and $L=\{L(v)\}_{v\in V}$. 

Let $v\in V$. For $u\in N_G(v)$ and colors $i,j\in L(v)$ with $i\neq j$, we denote by $(G_v, L_u^{ij})$ the list-coloring instance with $G_v=G\bs v$ and lists $L_u^{ij} = \{L_u^{ij}(w)\}_{w\in V\bs \{v\}}$ obtained from $L$ by:
\begin{itemize}[itemsep=-0.1cm,topsep=0.05cm]
\item removing the color $i$ from the lists $L(u')$ for $u'\in N_G(v)$ with $u' < u$, 
\item removing the color $j$ from the lists $L(u')$ for $u' \in N_G(v)$ with $u' > u$, and  
\item keeping the remaining lists unchanged.
\end{itemize}
\end{defn}
The following lemma will be crucial in our recursive approach to bound influences, and follows  by adapting suitably ideas from \cite{GKM}.
\newcommand{\statelemrecinf}{Let $(G,L)$ be a $(\Delta,q)$-list-coloring instance with $G=(V,E)$ and $L=\{L(v)\}_{v\in V}$. Then, for $v\in V$ and arbitrary colors $i,j \in L(v)$ with $i\neq j$, for all $w\in V\bs \{v\}$ and $k \in [q]$, we have 
\begin{align*}
\PP(\sigma_w \seq k \mid \sigma_v \seq i) 
&- 
\PP(\sigma_w \seq k \mid \sigma_v \seq j)=\\ 
&\sum_{u \in N_G(v)} 
\frac{\PP^{ij}_{u}(\sigma_u \seq j)}{\PP^{ij}_{u}(\sigma_u \sneq j)} \cdot \MM^{ij}_{u}\big((u,j), (w,k)\big)-\frac{\PP^{ij}_{u}(\sigma_u \seq i)}{\PP^{ij}_{u}(\sigma_u \sneq i)} \cdot \MM^{ij}_{u}\big((u,i), (w,k)\big),
\end{align*}
where $\PP:=\PP_{G,L}$ and, for $u\in N_G(v)$, $\PP^{ij}_{u}:=\PP_{G_v,L_u^{ij}}$ and $\MM^{ij}_{u}:=\MM_{G_v,L_u^{ij}}$.
}
\begin{lem}\label{lem:rec-inf}
\statelemrecinf
\end{lem}
Recall that we set $\MM^{ij}_{u}\big((u,c), (w,k)\big) = 0$ for $c \notin L_u^{ij}(u)$ (see \cref{foot:ext}). 
To apply~\cref{lem:rec-inf} recursively, it will be helpful to consider multiple list-coloring instances on the same graph $G$. 
For a collection of lists $\LL=\{L_1,\hdots, L_t\}$, where each $L\in \LL$ is a set of lists of all vertices for $G$, we use $(G,\LL)$ to denote the collection of $|\LL|$ list-coloring  instances $\{(G,L_1),\hdots, (G,L_t)\}$. 
When considering the pair $(G,\LL)$ or $(G,L)$, we usually omit the graph $G$ when it is clear from the context. 

\begin{defn}\label{def:GvLv}
Let $(G,\LL)$ be a collection of list-colorings instances with $G=(V,E)$ and a collection of lists $\LL$ on $G$. For $v\in V$, we define $\LL_v$ to be the collection of lists for $G_v=G\bs v$ obtained from $\LL$ by setting
\[\LL_v = \big\{L^{ij}_u\mid L \in \LL, u\in N_G(v), i,j \in L(v) \mbox{ with } i\neq j\big\}. \]
Note that $(G_v,\LL_v)$ consists of $|\LL_v| = \sum_{L \in \LL} \Delta_G(v) \cdot |L(v)| \cdot (|L(v)|-1)$ list-coloring instances. 
\end{defn}

\begin{lem}\label{claim:induced}
If $(G,\LL)$ is a collection of $(\Delta,q)$-list-coloring instances, then for every vertex $v$ of~$G$, $(G_v,\LL_v)$ is also a collection of $(\Delta,q)$-list-coloring instances.
\end{lem}
\begin{proof}
Let $L_v \in \LL_v$ be arbitrary, so that $L_v$ is obtained from some $L\in \LL$. Then, by definition, for $u \notin N_G(v)$ we have $|L_v(u)| = |L(u)|$ and $\Delta_{G\bs v}(u) = \Delta_G(u)$, while for $u \in N_G(v)$ we have $|L_v(u)| \ge |L(u)|-1$ and $\Delta_{G_v}(u) = \Delta_G(u)-1$. 
This implies that $\LL_v$ is $(\Delta,q)$-induced.  
\end{proof}

\subsection{Aggregating influences}
\label{sec:simpler-bound}
\begin{defn}\label{def:maxinf}
Let $(G,\LL)$ be a collection of $(\Delta,q)$-list-coloring instances with $G=(V,E)$. Fix a vertex $v\in V$ and let $w\in V\backslash\{v\}$, $k\in [q]$. The maximum influence of $v$ on $(w,k)$ with respect to $(G,\LL)$ is defined to be
\begin{equation*}
\II_{G,\LL}[v \sra (w,k)] 
= \max_{L\in \LL} \;\II_{G,L}[v \sra (w,k)]. 
\end{equation*}
The \emph{total} maximum influence of $v$ with respect to $(G,\LL)$ is defined to be 0 if  $\Delta_G(v) = 0$, and
\begin{equation*}
\II^*_{G,\LL}(v) = \frac{1}{\Delta_G(v)} \sum_{w\in V \bs \{v\}} \sum_{k\in [q]} \II_{G,\LL}[v \sra (w,k)]\quad \mbox{ if $\Delta_G(v)\geq1$}.
\end{equation*}
\end{defn}

The following lemma gives a recursive bound on the total maximum influence. 
\begin{lem}\label{lem:inf-rec-bound}
Let $(G,\LL)$ be a collection of list-coloring instances and $v$ be a vertex of $G$ with $\Delta_G(v)\geq 1$. Then, with $G_v,\LL_v$ as in~\cref{def:GvLv},
\[
\II^*_{G,\LL}(v) \le 
\max_{u\in N_G(v)} \Big\{ R_{G_v,\LL_v}(u)  \big( \Delta_{G_v}(u) \cdot \II^*_{G_v,\LL_v}(u) + q \big) \Big\},
\]
where $R_{G_v,\LL_v}(u) = \max_{L \in \LL_v} \max_{c\in L(u)} 
\frac{\PP_{G_v,L}(\sigma_u \seq c)}{\PP_{G_v,L}(\sigma_u \sneq c)}$ for $u\in N_G(v)$.
\end{lem}
\begin{proof}
Suppose that $G=(V,E)$. For convenience, we will drop the subscripts $G,\LL$ from influences  and use the subscript $v$ as a shorthand for the subscripts $G_v,\LL_v$ of influences and the quantity $R$. We will soon show that for every  $w\in V\bs \{v\}$ and color $k \in [q]$, we have 
\begin{equation}\label{eq:vwk2434}
\II[v \sra (w,k)] \le 
\sum_{u\in N_G(v)} R_{v}(u) \cdot \II_{v}[u \sra (w,k)]. 
\end{equation}
Assuming \eqref{eq:vwk2434} for the moment, we have that
\begin{align*}
\II^*(v) 
&= \frac{1}{\Delta_G(v)} \sum_{w\in V \bs \{v\}}\sum_{k\in [q]}  \II[v \sra (w,k)]
\le \frac{1}{\Delta_G(v)} \sum_{w\in V \bs \{v\}} \sum_{k\in [q]} \sum_{u\in N_G(v)} R_{v}(u) \cdot \II_{v} [u \sra (w,k)]\\ 
&= \frac{1}{\Delta_G(v)} \sum_{u\in N_G(v)} R_{v}(u) \cdot 
\bigg(
\sum_{w\in V \bs \{v,u\}} \sum_{k\in [q]} \II_{v} [u \sra (w,k)] 
+ 
\sum_{k\in [q]} \II_{v} [u \sra (u,k)]
\bigg)\\
&\le \max_{u\in N_G(v)} \Big\{ R_{v}(u)  \big( \Delta_{G_v}(u) \cdot \II^*_{v}(u) + q \big) \Big\}, 
\end{align*}
which is precisely the desired inequality. To prove \eqref{eq:vwk2434}, consider $L\in \LL$ and $i,j \in L(v)$ with $i\neq j$. For simplicity, let $\PP:=\PP_{G,L}$ and, for $u\in N_G(v)$, $\PP^{ij}_{u}:=\PP_{G_v,L_u^{ij}}$,  $\MM^{ij}_{u}:=\MM_{G_v,L_u^{ij}}$, and $\II^{ij}_u=\II_{G_v,L_u^{ij}}$. Let also $P^{ij}_{w,k}:=\PP(\sigma_w \seq k \mid \sigma_v \seq i) - \PP(\sigma_w \seq k \mid \sigma_v \seq j)$, so that from \cref{lem:rec-inf} we have
\begin{equation}\label{eq:5rvrfv56}
P^{ij}_{w,k}=\sum_{u \in N_G(v)} 
\frac{\PP^{ij}_{u}(\sigma_u \seq j)}{\PP^{ij}_{u}(\sigma_u \sneq j)} \cdot \MM^{ij}_{u}\big((u,j), (w,k)\big)-\frac{\PP^{ij}_{u}(\sigma_u \seq i)}{\PP^{ij}_{u}(\sigma_u \sneq i)} \cdot \MM^{ij}_{u}\big((u,i), (w,k)\big).
\end{equation}
By the law of total probability, we have
\begin{align*}
\sum_{c \in L_u^{ij}(u)} \PP^{ij}_{u}(\sigma_u \seq c) \cdot \MM^{ij}_{u}\big((u,c),(w,k)\big)=\sum_{c \in L_u^{ij}(u)} \PP^{ij}_{u}(\sigma_u \seq c)  \big(\PP^{ij}_{u}(\sigma_w \seq k \mid \sigma_u \seq c) - \PP^{ij}_{u} (\sigma_w \seq k) \big)= 0; 
\end{align*}
so we conclude that
\begin{equation}\label{eq:5rvrfv56b}
m^{ij}_u:=\min_{i' \in L_u^{ij}(u)} \MM^{ij}_u\big((u,i'),(w,k)\big) \le 0
\mbox{~~and~~} 
M^{ij}_u:=\max_{j' \in L_u^{ij}(u)} \MM^{ij}_{u}\big((u,j'),(w,k)\big) \ge 0.
\end{equation}
Observe further that 
\begin{equation}\label{eq:5rvrfv56c}
\II^{ij}_u[u \sra (w,k)]=\max_{i',j'\in L^{ij}_u(u)} 
\big| \PP^{ij}_u (\sigma_w \seq k \mid \sigma_u \seq i') 
- 
\PP^{ij}_u (\sigma_w \seq k \mid \sigma_v \seq j') \big|=M^{ij}_u-m^{ij}_u.
\end{equation}
Combining \eqref{eq:5rvrfv56}, \eqref{eq:5rvrfv56b}, \eqref{eq:5rvrfv56c} we obtain that
\begin{align*}
P^{ij}_{w,k}&\le \sum_{u \in N_G(v)} R_v(u)  \big(M^{ij}_u-m^{ij}_u\big)= \sum_{u \in N_G(v)} 
R_{v}(u) \cdot \II^{ij}_u[u \sra (w,k)]\le \sum_{u \in N_G(v)} R_{v}(u) \cdot \II_{v}[u \sra (w,k)].
\end{align*}
Since $\II_{G,L}[v \sra (w,k)]=\max_{i,j\in L(v)}P^{ij}_{w,k}$,  by taking maximum over $i,j \in L(v)$ of the left-hand side, we obtain the same upper for  $\II_{G,L}[v \sra (w,k)]$. 
We then obtain \eqref{eq:vwk2434} by taking maximum over $L\in\LL$, and thus finish the proof. 
\end{proof}

For the bound in \cref{lem:inf-rec-bound} to be useful, we need to show that the ratio $R(u)$ defined there is strictly less than $1/\Delta_G(u)$. The following lemma does this for $(\Delta,q) \in \Lambda_\eps$, building on ideas from \cite{GMP,GKM}. \footnote{\label{foot:GMP}We remark that our region $\Lambda_\eps$ is slightly smaller than that of \cite{GMP}, where similar bounds are shown for $q \ge \alpha \Delta -\gamma$ for $\gamma\approx 0.4703$. The difference is that the arguments in \cite{GMP}  upper-bound $\PP_L(\sigma_u \seq c)$ instead of the ratio $\PP_L(\sigma_u \seq c)/\PP_L(\sigma_u \sneq c)$ which is relevant here, and which is clearly larger than $\PP_L(\sigma_u \seq c)$. See also the discussion before the upcoming~\cref{lem:Phi-bound}.}
\newcommand{\statelemboundratio}{Let $\eps > 0$ and $(\Delta,q) \in \Lambda_\eps$. 
Let $(G,L)$ be a $(\Delta,q)$-list-coloring instance with $G$ a triangle-free graph.
Then for every vertex $u$ of $G$ with degree at most $\Delta-1$ and every color $c\in L(u)$, we have 
\[
\frac{\PP_{G,L}(\sigma_u \seq c)}{\PP_{G,L}(\sigma_u \sneq c)} 
\le 
\min \left\{ 
	\frac{1}{(1+\eps) \Delta_G(u)},\; 
	\frac{4}{q} 
\right\}. 
\]}
\begin{lem}\label{lem:bound-ratio}
\statelemboundratio 
\end{lem}

We remark that when $\Delta_G(u)$ is small, the bound $1/\Delta_G(u)$ is poor and we shall apply the simpler crude bound $4/q$. 
The proof of \cref{lem:bound-ratio} can be found in \cref{sec:margin}.  Combining \cref{lem:inf-rec-bound,lem:bound-ratio}, we can now bound the total influence. 
\begin{thm}\label{thm:bound-infl}
Let $\eps > 0$ and $(\Delta,q) \in \Lambda_\eps$. 
Suppose that $(G,\LL)$ is a collection of $(\Delta,q)$-list-coloring instances where $G$ is a triangle-free graph. 
Then for every vertex $v$ of $G$ we have $\II^*_{G,\LL}(v) \le 4\big( \frac{1}{\eps} + 1 \big)$. 
\end{thm}

\begin{proof}
Let $v_0 = v$, $G^{0}=G$ and $\LL^0 = \LL$.  For $\ell \ge 0$, we will define inductively a sequence of $(\Delta,q)$-list-coloring instances $(G^{\ell},\LL^{\ell})$ and a vertex $v_{\ell}$ in $G^{\ell}$  as follows. Let $G^{\ell+1}$ be the graph obtained from $G^{\ell}$ by deleting $v_{\ell}$, i.e., $G^{\ell+1}=G^{\ell}\backslash v_{\ell}$ and $\LL^{\ell+1}=\LL^{\ell}_{v_\ell}$. Note that all neighbors of $v_{\ell}$ in $G_{\ell}$ have degree at most $\Delta-1$ in $G^{\ell+1}$. Moreover, since by induction $(G^{\ell},\LL^{\ell})$ is a set of $(\Delta,q)$-list-coloring instances, by~\cref{claim:induced} so is $(G^{\ell+1},\LL^{\ell+1})$. Since $q\geq (1+\eps)\alpha\Delta+1$,  combining \cref{lem:inf-rec-bound,lem:bound-ratio}, we obtain that
\begin{equation}\label{eq:I*-rec-bound}
\II^*_{G^{\ell},\LL^{\ell}}(v_{\ell}) 
\le 
\frac{1}{1+\eps} \cdot \max_{u\in N_{G^{\ell}}(v_{\ell})} \left\{ \II^*_{G^{\ell+1},\LL^{\ell+1}}(u) \right\} 
+ 
4. 
\end{equation}
  We  let $v_{\ell+1}$ be the vertex $u\in N_{G^{\ell}}(v_{\ell})$ that attains the maximum of the right-hand side of~\eqref{eq:I*-rec-bound}, so 
\begin{equation}\label{eq:proof-I*}
\II^*_{G^{\ell},\LL^{\ell}}(v_{\ell})  
\le \frac{1}{1+\eps} \cdot \II^*_{G^{\ell+1},\LL^{\ell+1}}(v_{\ell+1}) 
+ 4. 
\end{equation}
Hence, we obtain a sequence of vertices $v_0,v_1,\dots,v_m$ and collections of lists $\LL^0,\LL^1,\dots,\LL^m$, till when $\Delta_{G^m}(v_m) = 0$ and thus $\II_{G^m,\LL^m}^*(v_m) = 0$. From this, and since \eqref{eq:proof-I*} holds for all $0\le \ell \le m-1$, we obtain by solving the recursion that  
$\II_{G,\LL}^*(v) \le 
\frac{4}{1 - (1+\eps)^{-1}} 
= 4\left( \frac{1}{\eps} + 1 \right)$, as wanted. 
\end{proof}
Combining \cref{thm:bound-infl} with~\cref{obs:bound} and~\cref{def:maxinf} of total maximum influence  gives \eqref{eq:suma}, which therefore yields the bound $\lambda_1(\MM_{G,L})\leq 4\left( \frac{1}{\eps}+1 \right)\Delta$ for any $(\Delta,q)$-list-coloring instance $(G,L)$ with $(\Delta,q)\in \Lambda_\eps$, as claimed at the beginning of this section.

\subsection{An example where this spectral bound is not tight}
\label{sec:nottight}
From the arguments of the previous section we get that, for a $(\Delta,q)$-list-coloring instance $(G,L)$ with $(\Delta,q)\in \Lambda_\eps$ it holds that $\lambda_1(\MM_{G,L})\leq 4\left( \frac{1}{\eps}+1 \right)\Delta$. As discussed earlier, this only yields an $n^{C\Delta}$ upper bound on the mixing time for some $C=C(\alpha)>0$, which is exponential in the maximum degree $\Delta$. The following example shows that \eqref{eq:suma} and threfore the bound on $\lambda_1(\MM_{G,L})$  are not tight.  

\begin{eg}\label{eg:star}
Consider $q$-colorings of a star graph $G=(V,E)$ on $\Delta+1$ vertices centered at $v$. 
Then for every $w\in N_G(v) = V \bs \{v\}$ and every $k\in [q]$, we have $\II_{G,L}[v \sra (w,k)] = \frac{1}{q-1}$, and hence, 
\begin{equation}\label{eq:4f3ffwfwe}
\sum_{w \in V \bs \{v\}} \sum_{k\in [q]} \II_{G,L}[v \sra (w,k)] = \frac{q}{q-1} \cdot \Delta \geq \Delta. 
\end{equation}
Meanwhile, given $i\in [q]$, for every $w\in N_G(v) = V \bs \{v\}$ and every $k\in [q]$ we have
\[
\MM_{G,L}((v,i),(w,k)) = \frac{1}{q(q-1)} \mbox{~~if~~} k\neq i, \quad \MM_{G,L}((v,i),(w,k)) =-\frac{1}{q} \mbox{~~if~~} k= i.
\]
Therefore, for every $i\in[q]$ we have $\sum_{w \in V \bs \{v\}} \sum_{k\in [q]} \Big| \MM_{G,L}((v,i),(w,k)) \Big| = \frac{2\Delta}{q}$, which is a factor of at least $q/2$ smaller than the bound in \eqref{eq:4f3ffwfwe}.
\end{eg}
\cref{eg:star} indicates that the maximum influence $\II_L[v \sra (w,k)]$ does not always provide a good bound on $\MM_L((v,i),(w,k))$; in fact, as we will see in the next section in detail, it loses a factor of roughly $q$ when it comes to the off-diagonal entries, i.e., when $k \neq i$.

\section{Polynomial mixing time for all \texorpdfstring{$\Delta$}{Delta}}
\label{sec:better-bound}

In this section, we prove the constant upper bound on the largest eigenvalue of $\MM_{G,L}$ for list-coloring instances $(G,L)$ as  in~\cref{thm:lambda1}. To tighten the analysis of the previous section and motivated from  the bad example of~\cref{sec:nottight}, we introduce the \emph{maximum biased influence} which describes the maximum difference of the marginal probability of $\sigma_w = k$ under ``non-$k$'' color choices of $v$. 

\begin{defn}
Let $(G,\LL)$ be a collection of $(\Delta,q)$-list-coloring instances with $G=(V,E)$. Fix a vertex $v\in V$, and let $w\in V$ and $k\in [q]$.  For $L\in \LL$, the \emph{maximum biased influence} of $v$ on $(w,k)$ with respect to $(G,L)$ is defined as
\[
\hat{\II}_{G,L}[v \sra (w,k)] 
= 
\max_{i,j\in L(v) \bs \{k\}} 
\big| \PP_{G,L} (\sigma_w \seq k \mid \sigma_v \seq i) 
- 
\PP_L (\sigma_w \seq k \mid \sigma_v \seq j) \big|. 
\] 
The maximum biased influence of $v$ on $(w,k)$ with respect to $(G,\LL)$ is defined to be $\hat{\II}_{G,\LL}[v \sra (w,k)] = \max_{L \in \LL}\, \hat{\II}_{G,L}[v \sra (w,k)]$.  The \emph{total maximum biased influence} of $v$ with respect to $(G,\LL)$ is defined to be 0 if $\Delta_G(v) = 0$, and  
\[
\hat{\II}^*_{G,\LL}(v) = \frac{1}{\Delta_G(v)} \sum_{w\in V \bs \{v\}} \sum_{k\in [q]} \hat{\II}_{G,\LL}[v \sra (w,k)], \mbox{\ \ \ if $\Delta_G(v) \ge 1$}.
\]
\end{defn}

We can upper bound $\sum_{w \in V \bs \{v\}} \sum_{k\in [q]} |\MM_L((v,i),(w,k))|$ by a weighted sum of $\II^*_L(v)$ and $\hat{\II}^*_L(v)$, and from that we are able to get a more precise bound, saving a factor of $q$. 

\newcommand{\statelemMleq}{Let $(G,L)$ be a $(\Delta,q)$-list-coloring instance with $G=(V,E)$. For $v\in V$ and $i\in L(v)$, we have
\[
\sum_{w \in V \bs \{v\}} \sum_{k\in [q]} \Big| \MM_{G,L} ((v,i),(w,k)) \Big| 
\le 
2 \Delta_G(v) \left( \hat{\II}^*_{G,L}(v) + P_{G,L}(v) \cdot \II^*_{G,L}(v) \right) 
\]
where $P_{G,L}(v) = \max_{c\in L(v)} \PP_{G,L}(\sigma_v \seq c)$.
}
\begin{lem}\label{lem:M<=I+hatI}
\statelemMleq
\end{lem}

\cref{lem:M<=I+hatI} is proved by applying the law of total probability to the left-hand side and bounding each term with either $\II_L[\cdot]$ or $\hat{\II}_L[\cdot]$ respectively; the proof  can be found in \cref{subsec:pf-M}. 
Since we know that $P_{G,L}(v) = O(1/q)$ from \cref{lem:bound-ratio} and also $\II^*_{G,L}(v) = O(1)$ from \cref{thm:bound-infl}, it suffices to show $\hat{\II}^*_{G,L}(v) = O(1/q)$ in order to get a bound of $O(\Delta/q)$ for the row sums of $\MM_{G,L}$. 
The remaining of this section aims to prove this. We first give a recursive upper bound on the total maximum biased influence, which can be viewed as an analogue of \cref{lem:inf-rec-bound} for maximum biased influences. 
\newcommand{\statelemhatinf}{Let $(G,\LL)$ be a collection of list-coloring instances and $v$ be a vertex of $G$ with $\Delta_G(v)\geq 1$. Then, with $G_v,\LL_v$ as in~\cref{def:GvLv},
\[
\hat{\II}^*_{G,\LL}(v) \le 
\max_{u\in N_G(v)} 
\left\{
	R_{G_v,\LL_v}(u) \cdot 
	\left[ 
		\Delta_{G\bs v}(u) \cdot \hat{\II}^*_{G_v, \LL_v}(u) 
		+ 
		R_{G_v, L_v}(u) \cdot \left( \Delta_{G_v}(u) \cdot \II^*_{G_v,\LL_v}(u) + q\right) 
	\right]
\right\},
\]
where 
$R_{G_v,\LL_v}(u) = \max_{L \in \LL_v} \max_{c\in L(u)} 
\frac{\PP_{G_v,L}(\sigma_u \seq c)}{\PP_{G_v,L}(\sigma_u \sneq c)}$ for $u \in N_G(v)$.}
\begin{lem}\label{lem:hatinf-rec-bound}
\statelemhatinf
\end{lem}

Notice that the right-hand side of inequality in the lemma includes both the influence $\II^*_{G_v,\LL_v}(u)$ and the biased influence $\hat{\II}^*_{G_v,\LL_v}(u)$. 
Combining \cref{lem:bound-ratio} and \cref{thm:bound-infl}, we obtain the following. 

\begin{thm}\label{thm:bound-hatinfl}
Let $\eps > 0$. 
Let $(G,\LL)$ be a collection of $(\Delta,q)$-list-coloring instances where $G$ is a triangle-free graph and $(\Delta,q) \in \Lambda_\eps$,
Then, for every vertex $v$ of $G$, and with $(G_v,\LL_v)$ as in~\cref{def:GvLv}, we have 
\begin{equation}\label{eq:hatI*-rec-bound}
\hat{\II}^*_{G,\LL}(v) \le \frac{1}{1+\eps} \cdot \max_{u\in N_G(v)} \left\{ \hat{\II}^*_{G_v,\LL_v}(u) \right\} 
+ \frac{16}{q} \left( \frac{1}{\eps}+1 \right).  
\end{equation}
Therefore, $\hat{\II}^*_{G,\LL}(v) \le \frac{16}{q} \big( \frac{1}{\eps}+1 \big)^2$. 
\end{thm}
\begin{proof}
To prove \eqref{eq:hatI*-rec-bound}, we bound each of the terms appearing in the maximization for $\hat{\II}^*_{G,\LL}(v)$ in \cref{lem:hatinf-rec-bound}. By \cref{lem:bound-ratio}, for every $u\in N_G(v)$, we have that $R_{G_v,\LL_v}(u)\leq \frac{4}{q}$ and $R_{G_v,\LL_v}(u) \cdot \Delta_{G_v}(u)\leq \frac{1}{1+\eps}$ and by \cref{thm:bound-infl} we have that $\II^*_{G,\LL}(v) \le 4\big( \frac{1}{\eps} + 1 \big)$. Therefore, the bound in \cref{lem:hatinf-rec-bound} gives
\begin{align*}
\hat{\II}^*_{G,\LL}(v) 
&\le \max_{u\in N_G(v)} 
\left\{
	\frac{1}{1+\eps} \cdot \hat{\II}^*_{G_v,\LL_v}(u) 
	+ 
	\frac{4}{q} \cdot \frac{1}{1+\eps} \cdot 4\left( \frac{1}{\eps} + 1 \right) 
	+
	\frac{16}{q^2} \cdot q
\right\}\\
&= \frac{1}{1+\eps} \cdot \max_{u\in N_G(v)} \left\{ \hat{\II}^*_{G_v,\LL_v}(u) \right\} + \frac{16}{q} \left( \frac{1}{\eps} + 1 \right). 
\end{align*}
This establishes \eqref{eq:hatI*-rec-bound}. From this, $\hat{\II}^*_{G,\LL}(v) \le \frac{16}{q} \big( \frac{1}{\eps}+1 \big)^2$ is obtained analogously to \cref{thm:bound-infl}, see \eqref{eq:I*-rec-bound} and \eqref{eq:proof-I*}, by solving the recursion induced by \eqref{eq:hatI*-rec-bound}.   
\end{proof}

We are now ready to prove~\cref{thm:lambda1} which we restate here for convenience.
\begin{thmlambda}
\statethmlambda
\end{thmlambda}
\begin{proof}
From \cref{lem:M<=I+hatI}, \cref{thm:bound-infl}, and \cref{thm:bound-hatinfl} we get for every vertex $v$ of $G$ and every $i\in L(v)$ that
\[
\sum_{w \in V \bs \{v\}} \sum_{k\in [q]} \Big| \MM_{G,L} ((v,i),(w,k)) \Big| 
\le 64 \left( \frac{1}{\eps}+1 \right)^2 \frac{\Delta}{q}. 
\]
This implies that the row sums of $\MM_{G,L}$ are bounded by the same quantity, yielding therefore the desired bound on $\lambda_1(\MM_{G,L})$. Note, the bound is tight in $\Delta$ and $q$ as illustrated in \cref{eg:star}. 
\end{proof}

\subsection{Proof of \texorpdfstring{\cref{lem:M<=I+hatI,lem:hatinf-rec-bound}}{Lemmas 23 and 24}}

In this section, we give the proof of the remaining \cref{lem:M<=I+hatI,lem:hatinf-rec-bound} that were used in the proof of \cref{thm:lambda1}. Let $(G,L)$ be a $(\Delta,q)$-list-coloring instance with $G=(V,E)$. Given $v,w\in V$ and $k\in [q]$, it will be helpful to define
\[
\hat{\JJ}_{G,L}[v \sra (w,k)] 
= 
\max_{i\in L(v) \bs \{k\}} 
\left| \PP_{G,L} (\sigma_w \seq k \mid \sigma_v \seq i) 
- 
\PP_{G,L} (\sigma_w \seq k) \right|. 
\]
The quantity $\hat{\JJ}_{G,L}[v \sra (w,k)]$ is upper bounded by a weighted sum of $\II_{G,L}[v \sra (w,k)]$ and $\hat{\II}_{G,L}[v \sra (w,k)]$, as shown by the following lemma.
\begin{lem}\label{claim:J-K}
Fix an arbitrary vertex $v\in V$. 
Let $w\in V$ and $k \in [q]$. 
Then we have 
\[
\hat{\JJ}_{G,L}[v \sra (w,k)] \le 
(1-P_{G,L}(v)) \cdot \hat{\II}_{G,L}[v \sra (w,k)] + P_{G,L}(v) \cdot \II_{G,L}[v \sra (w,k)], 
\]
where $P_{G,L}(v) = \max_{c\in L(v)} \PP_{G,L}(\sigma_v \seq c)$. 
\end{lem}
\begin{proof}
For convenience, we drop the subscript $G$ from notation. For $i\in L(v) \bs \{k\}$, using the law of total probability and the triangle inequality we have
\begin{align*}
 \big| \PP_{L} (\sigma_w \seq k \mid \sigma_v \seq i) 
- 
\PP_{L} (\sigma_w \seq k) \big|& \leq \sum_{j \in L(v)} \PP_{L}(\sigma_v \seq j) \cdot 
\big| \PP_L (\sigma_w \seq k \mid \sigma_v \seq i) 
- 
\PP_L (\sigma_w \seq k \mid \sigma_v \seq j) \big|\\ 
&\leq 
\big( 1-\PP_L(\sigma_v \seq k) \big) \cdot \hat{\II}_L[v \sra (w,k)] 
+ 
\PP_L(\sigma_v \seq k) \cdot \II_L[v \sra (w,k)],
\end{align*}
where the last inequality follows from the definitions of $\hat{\II}_L[v \sra (w,k)]$ and $\II_L[v \sra (w,k)]$. 
Now since $\hat{\II}_L[v \sra (w,k)] \le \II_L[v \sra (w,k)]$ and $\PP_L(\sigma_v \seq k) \le P_L(v)$, we deduce that
\[
\big| \PP_L (\sigma_w \seq k \mid \sigma_v \seq i) 
- 
\PP_L (\sigma_w \seq k) \big|
\le (1-P_L(v)) \cdot \hat{\II}_L[v \sra (w,k)] 
+ P_L(v) \cdot \II_L[v \sra (w,k)]. 
\]
The lemma then follows by taking maximum over $i\in L(v) \bs \{k\}$ on the left-hand side. 
\end{proof}

\subsubsection{Proof of \texorpdfstring{\cref{lem:M<=I+hatI}}{Lemma 23}}
\label{subsec:pf-M}
We are now ready to prove Lemma~\cref{lem:M<=I+hatI}, which we restate here for convenience.
\begin{lemMleq}
\statelemMleq
\end{lemMleq}
\begin{proof}
For convenience, we drop the subscript $G$ from notation. We consider separately the terms where $k\neq i$ and $k=i$.  By \cref{claim:J-K}, we get
\begin{align}
\sum_{w \in V \bs \{v\}} \sum_{k\in [q] \bs \{i\}} \Big| \MM_L \big((v,i),(w,k)\big) \Big|&\leq  
\sum_{w \in V \bs \{v\}} \sum_{k\in [q] \bs \{i\}} 
\hat{\JJ}_L[v \to (w,k)]\notag\\
& \leq
\sum_{w \in V \bs \{v\}} \sum_{k\in [q] \bs \{i\}} 
\left( \hat{\II}_L[v \sra (w,k)] + P_L(v) \cdot \II_L[v \sra (w,k)] \right)\notag\\
&\leq \Delta_G(v) \left( \hat{\II}_L^*(v) + P_L(v) \cdot \II_L^*(v) \right). \label{eq:4f45f5f5rfa}
\end{align}
Note that 
\[\MM_L \big((v,i),(w,i)\big)=\PP_L (\sigma_w \seq i \mid \sigma_v \seq i) - \PP_L (\sigma_w \seq i)=
- \sum_{k \in [q] \bs \{i\}} 
\Big( \PP_L (\sigma_w \seq k \mid \sigma_v \seq i) - \PP_L (\sigma_w \seq k) \Big)\]
and hence using the triangle inequality we obtain that 
\begin{align}
\sum_{w \in V \bs \{v\}} \Big| \MM_L \big((v,i),(w,i)\big) \Big| &\le \sum_{w \in V \bs \{v\}} \sum_{k \in [q] \bs \{i\}} \big| \PP_L (\sigma_w \seq k \mid \sigma_v \seq i) - \PP_L (\sigma_w \seq k) \big|\notag\\ 
&= \sum_{w \in V \bs \{v\}} \sum_{k\in [q] \bs \{i\}} \Big| \MM_{L} \big((v,i),(w,k)\big) \Big|\leq  
\sum_{w \in V \bs \{v\}} \sum_{k\in [q] \bs \{i\}} 
\hat{\JJ}_L[v \to (w,k)]\notag\\ 
&\le \Delta_G(v) \left( \hat{\II}_L^*(v) + P_L(v) \cdot \II_L^*(v) \right). \label{eq:4f45f5f5rfab}
\end{align}
The lemma then follows by adding \eqref{eq:4f45f5f5rfa} and \eqref{eq:4f45f5f5rfab}. 
\end{proof}

\subsubsection{Proof of \texorpdfstring{\cref{lem:hatinf-rec-bound}}{Lemma 24}}
We now prove \cref{lem:hatinf-rec-bound}. First, we establish a recursive inequality for the maximum biased influence on a specific pair $(w,k)$.
\begin{lem}\label{lem:hatinf(wk)-rec-bound}
Let $(G,\LL)$ be a collection of $(\Delta,q)$-list-coloring instances with $G=(V,E)$. Fix an arbitrary vertex $v\in V$ and let $w\in V\backslash \{v\}$ and $k \in [q]$. 
Then, with $G_v,\LL_v$ as in~\cref{def:GvLv}, we have 
\[
\hat{\II}_{G,\LL}[v \sra (w,k)] \le 
\sum_{u\in N_G(v)} R_{G_v,\LL_v}(u) \cdot 
\left( \hat{\II}_{G_v,\LL_v}[u \sra (w,k)] 
+ 
R_{G_v,\LL_v}(u) \cdot \II_{G_v,\LL_v}[u \sra (w,k)] \right);
\]
where $R_{G_v,\LL_v}(u) = \max_{L \in \LL_v} \max_{c\in L(u)} 
\frac{\PP_{G_v,L}(\sigma_u \seq c)}{\PP_{G_v,L}(\sigma_u \sneq c)}$.
\end{lem}
\begin{proof}
Let $L\in \LL$ and $i,j \in L(v) \bs \{k\}$.  For simplicity, we will use the shorthands $\PP:=\PP_{G,L}$ and, for $u\in N_G(v)$, 
\[\PP^{ij}_{u}:=\PP_{G_v,L_u^{ij}},\quad  \MM^{ij}_{u}:=\MM_{G_v,L_u^{ij}},\quad \hat{\II}^{ij}_u:=\hat{\II}_{G_v,L_u^{ij}},\quad R_v(u):=R_{G_v,\LL_v}(u).\]
Let also $P^{ij}_{w,k}:=\PP(\sigma_w \seq k \mid \sigma_v \seq i) - 
\PP(\sigma_w \seq k \mid \sigma_v \seq j)$, so that from \cref{lem:rec-inf} we have
\begin{equation*}
P^{ij}_{w,k}=\sum_{u \in N_G(v)} 
\frac{\PP^{ij}_{u}(\sigma_u \seq j)}{\PP^{ij}_{u}(\sigma_u \sneq j)} \cdot \MM^{ij}_{u}\big((u,j), (w,k)\big)-\frac{\PP^{ij}_{u}(\sigma_u \seq i)}{\PP^{ij}_{u}(\sigma_u \sneq i)} \cdot \MM^{ij}_{u}\big((u,i), (w,k)\big).
\end{equation*}
Define $x^+ = \max\{x,0\}$ and $x^- = -\min\{x,0\}$ for $x \in \R$. From \cref{lem:rec-inf} we have
\begin{align*}
P^{ij}_{w,k}&\leq \sum_{u \in N_G(v)} 
R_v(u) \cdot 
\left[ 
\left( \MM^{ij}_u\big((u,j),(w,k)\big) \right)^+ +  \left( \MM^{ij}_u\big((u,i),(w,k)\big) \right)^-\right]\\ 
& \leq 
\sum_{u \in N_G(v)} 
R_v(u) \cdot 
\max_{i',j' \in L(u) \bs \{k\}}
\left[ \left( \MM^{ij}_u\big((u,j'),(w,k)\big) \right)^+ + \left(\MM^{ij}_u\big((u,i'),(w,k)\big) \right)^- \right]\\ 
&\leq 
\sum_{u \in N_G(v)} 
R_v(u) \cdot 
\max \left\{ \hat{\II}^{ij}_u[u \sra (w,k)],\, \hat{\JJ}^{ij}_u[u \sra (w,k)] \right\}.
\end{align*}
By \cref{claim:J-K}, for $u\in N_G(v)$ we can bound $\hat{\JJ}^{ij}_u[u \sra (w,k)]$ by $\hat{\II}^{ij}_u[u \sra (w,k)] + R_v(u) \cdot \II^{ij}_u[u \sra (w,k)]$.  Therefore, we get
\[
P^{ij}_{w,k}
\leq 
 \sum_{u \in N_G(v)} 
R_v(u) \cdot \left( 
\hat{\II}^{ij}_u[u \sra (w,k)] + R(u) \cdot \II^{ij}_u[u \sra (w,k)] \right).
\]
Taking maximum over $L\in \LL$ and $i,j \in L(v) \bs \{k\}$, we obtain the lemma. 
\end{proof}

We then deduce \cref{lem:hatinf-rec-bound} from \cref{lem:hatinf(wk)-rec-bound}. 
\begin{lemhatinf}
\statelemhatinf
\end{lemhatinf}
\begin{proof}
For convenience, we will drop the subscripts $G,\LL$ from influences  and use the subscript $v$ as a shorthand for the subscripts $G_v,\LL_v$ of influences and the $R$-quantity. By \cref{lem:hatinf(wk)-rec-bound}, we have
\begin{align*}
\hat{\II}^*&(v) 
= \frac{1}{\Delta_G(v)} \sum_{w\in V \bs \{v\}} \sum_{k\in [q]} \hat{\II}[v \sra (w,k)]\\ 
&\le \frac{1}{\Delta_G(v)} \sum_{w\in V \bs \{v\}} \sum_{k\in [q]} 
\sum_{u\in N_G(v)} R_v(u) \cdot 
\left( \hat{\II}_{v}[u \sra (w,k)] 
+ 
R_{v}(u) \cdot \II_{v}[u \sra (w,k)] \right)\\ 
&= \frac{1}{\Delta_G(v)} \sum_{u\in N_G(v)} R_{v}(u) \cdot 
\bigg(
	\sum_{w\in V \bs \{v\}} \sum_{k\in [q]} \hat{\II}_{v}[u \sra (w,k)] 
	+ 
	R_{v}(u) \sum_{w\in V \bs \{v\}} \sum_{k\in [q]} \II_{v}[u \sra (w,k)]
\bigg).
\end{align*}
Now we have
\begin{align*}
\sum_{w\in V \bs \{v\}} \sum_{k\in [q]} \II_{v}[u \sra (w,k)] 
&= \sum_{w\in V \bs \{v,u\}} \sum_{k\in [q]} \II_{v}[u \sra (w,k)]
+ \sum_{k\in [q]} \II_{v}[u \sra (u,k)]\le \Delta_{G\bs v}(u) \cdot \II_{v}^*(u) + q.
\end{align*}
Observing further that 
$$
\hat{\II}_{v}[u \sra (u,k)] 
= \max_{L \in \LL_v} \max_{i,j\in L(u) \bs \{k\}} 
\left| \PP_{G_v,L} (\sigma_u \seq k \mid \sigma_u \seq i) 
- 
\PP_{G_v,L} (\sigma_u \seq k \mid \sigma_u \seq j) \right| = 0,
$$ 
we have 
\begin{align*}
\sum_{w\in V \bs \{v\}} \sum_{k\in [q]} \hat{\II}_{v}[u \sra (w,k)] 
&= \sum_{w\in V \bs \{v,u\}} \sum_{k\in [q]} \hat{\II}_{v}[u \sra (w,k)]
+ \sum_{k\in [q]} \hat{\II}_{v}[u \sra (u,k)]= \Delta_{G\bs v}(u) \cdot \hat{\II}_{v}^*(u).
\end{align*}
Hence, we deduce that
\begin{align*}
\hat{\II}_\LL^*(v) 
&\le \frac{1}{\Delta_G(v)} \sum_{u\in N_G(v)} R_{v}(u) \cdot 
\left[
	\Delta_{G_v}(u) \cdot \hat{\II}_{v}^*(u) 
	+ 
	R_{v}(u) \cdot \left( \Delta_{G_v}(u) \cdot \II_{v}^*(u) + q \right)
\right]\\
&\le \max_{u\in N_G(v)} 
\left\{ R_{v}(u) \cdot 
	\left[
		\Delta_{G_v}(u) \cdot \hat{\II}_{v}^*(u) 
		+ 
		R_{v}(u) \cdot \left( \Delta_{G_v}(u) \cdot \II_{v}^*(u) + q \right)
	\right]
\right\},
\end{align*}
which finishes the proof of the lemma. 
\end{proof}

\section{Remaining proofs: recursion and marginal bounds}
\label{sec:recboundmarg}
In this section, we give the proof of~\cref{lem:rec-inf} and~\cref{lem:bound-ratio}, which were used in the proof of~\cref{thm:lambda1}.
\subsection{Proof of \texorpdfstring{\cref{lem:rec-inf}}{Lemma 14}}
\label{sec:rec-inf}
In this section, we prove the recursion of \cref{lem:rec-inf} which we restate here for convenience.
\begin{lemrecinf}
\statelemrecinf
\end{lemrecinf}
\begin{proof}
For convenience, set $P:=\PP(\sigma_w \seq k \mid \sigma_v \seq j) -\PP(\sigma_w \seq k \mid \sigma_v \seq i)$.

Let $d=\Delta_G(v)$ and $u_1,\hdots, u_d$ be the neighbors of $v$ in $G$ in the order prescribed by the labelling on $G$. Let $N=N_G(v)$ and, for $t=1,\hdots,d$, let $N_{t}=\{u_1,\hdots u_{t-1}\}$ be the set of  vertices preceding $v_t$. Then, with $G_v=G\backslash v$ and $L_v=\{L(u)\}_{u\in V\bs\{v\}}$, we have
\begin{align*}
P&=\PP_{G_v,L_v}\big(\sigma_w \seq k,\, j\notin \sigma_N\big)  -\PP_{G_v,L_v}\big(\sigma_w \seq k,\, i\notin  \sigma_N\big)\\
&=\sum^{d}_{t=1}\PP_{G_v,L_v}\big(\sigma_w \seq k,\, i\notin \sigma_{N_{t}},\, j\notin \sigma_{N\bs N_t} \big)  -\PP_{G_v,L_v}\big(\sigma_w \seq k,\, i\notin \sigma_{N_{t+1}},\, j\notin \sigma_{N\bs N_{t+1}} \big)\\
&=\sum_{u\in N_G(v)} \PP^{ij}_u(\sigma_w \seq k \mid \sigma_u \sneq j) - \PP^{ij}_u(\sigma_w \seq k \mid \sigma_u \sneq i).  
\end{align*}
Now, for $u\in N_G(v)$, we have that
\begin{align*}
\PP^{ij}_u(\sigma_w \seq k \mid \sigma_u \sneq i) - \PP^{ij}_u(\sigma_w \seq k)=\begin{cases} 0,& \mbox{~if~} i\notin L(u),\\ -\displaystyle\frac{\PP^{ij}_u(\sigma_u \seq i)}{\PP^{ij}_u(\sigma_u \sneq i)} \cdot \MM^{ij}_u\big((u,i),(w,k)\big),&\mbox{~if~} i\in L(u).\end{cases} 
\end{align*}
Summing this over $u\in N_G(v)$ yields the equality in the lemma.
\end{proof}

\subsection{Bounding marginal probabilities}
\label{sec:margin}
In this section, we prove \cref{lem:bound-ratio}. For integers $\Delta, q \ge 3$ with $q \ge \Delta+1$, the following function will be relevant for this section:
\begin{equation}\label{eq:Phi}
\Phi(\Delta,q) = \frac{q-2}{\Delta-1} \cdot \bigg[ \Big( 1 - \frac{1}{q-\Delta+1} \Big)^{q-\Delta+1} \bigg]^{\frac{\Delta-1}{q-2}}. 
\end{equation}
The following lemma is implicitly given in \cite{GMP} in their proof of Lemma 15. 
Here we present a more direct proof, combining ideas from both \cite{GMP} and \cite{GKM}. 
\begin{lem}\label{lem:new-ratio-bound}
Suppose that $(G,L)$ is a $(\Delta,q)$-list-coloring instance with $G=(V,E)$ a triangle-free graph. 
Then for every vertex $u\in V$ of degree at most $\Delta-1$ and every color $c\in L(u)$, we have 
\[
\frac{\PP_{G,L}(\sigma_u \seq c)}{\PP_{G,L}(\sigma_u \sneq c)} 
\le \frac{1}{\Phi(\Delta,q)} \cdot \frac{1}{\Delta_G(u)}. 
\]
\end{lem}
\begin{proof}
By the law of total probability, it suffices to give an upper bound on $\frac{\PP_{G_\tau,L_\tau}(\sigma_u \seq c)}{\PP_{G_\tau,L_\tau}(\sigma_u \sneq c)}$ for an arbitrary partial list-coloring $\tau$ on $V\backslash ({u}\cup N_G(u))$. In turn, since $(G_\tau,L_\tau)$ is also a $(\Delta,q)$-list-coloring instance 
(see for example the argument in the proof of~\cref{thm:main} in~\cref{sec:proofmain}) 
and $G_\tau$ is a star graph centered at $u$, it suffices to prove the lemma when $G$ is a star graph centered at $u$. Henceforth, for convenience,  we drop the subscript $G$ from notation.

For $w \in N_G(u)$ and $c \in L(u)$ we define $\delta_c(w) = \one\{c \in L(w)\}$. 
For any $c,c' \in L(u)$, we have
\[
\frac{\PP_L(\sigma_u \seq c')}{\PP_L(\sigma_u \seq c)} = \prod_{w \in N_G(u)} \frac{|L(w)| - \delta_{c'}(w)}{|L(w)| - \delta_c(w)} 
\ge \prod_{w \in N_G(u)} \left( 1 - \frac{\delta_{c'}(w)}{|L(w)|} \right) 
= \prod_{w \in N_G(u)} \left( 1 - \frac{1}{|L(w)|} \right)^{\delta_{c'}(w)}. 
\]
From this, and using the arithmetic-geometric mean inequality, it follows that
\begin{align*}
\frac{\PP_L(\sigma_u \sneq c)}{\PP_L(\sigma_u \seq c)} 
&= \sum_{c' \in L(u) \bs \{c\}} \prod_{w \in N_G(u)} \bigg( 1 - \frac{1}{|L(w)|} \bigg)^{\delta_{c'}(w)}\\ 
&\ge \big(|L(u)|-1\big) \bigg( \prod_{c' \in L(u) \bs \{c\}} \prod_{w \in N_G(u)} \left( 1 - \frac{1}{|L(w)|} \right)^{\delta_{c'}(w)} \bigg)^{\frac{1}{|L(u)|-1}}\\
&= \big(|L(u)|-1\big) \bigg( \prod_{w \in N_G(u)} \left( 1 - \frac{1}{|L(w)|} \right)^{\sum_{c' \in L(u) \bs \{c\}} \delta_{c'}(w)} \bigg)^{\frac{1}{|L(u)|-1}}\\ 
&\ge \big(|L(u)|-1\big) \bigg( \prod_{w \in N_G(u)} \left( 1 - \frac{1}{|L(w)|} \right)^{|L(w)|} \bigg)^{\frac{1}{|L(u)|-1}}.
\end{align*}
Since $(1-1/m)^m$ is an increasing sequence in $m$ and $|L(w)| \ge q-\Delta+1$, we get 
\begin{equation*}
\frac{1}{\Delta_G(u)} \cdot \frac{\PP_L(\sigma_u \sneq c)}{\PP_L(\sigma_u \seq c)} 
\ge \frac{|L(u)|-1}{\Delta_G(u)} \cdot \left[ \left( 1 - \frac{1}{q-\Delta+1} \right)^{q-\Delta+1} \right]^{\frac{\Delta_G(u)}{|L(u)|-1}}. 
\end{equation*}
Since we have
\[
\frac{|L(u)|-1}{\Delta_G(u)} \ge \frac{q-\Delta-1}{\Delta_G(u)} + 1 \ge \frac{q-\Delta-1}{\Delta-1} + 1 = \frac{q-2}{\Delta-1}, 
\]
we deduce that
\[
\frac{1}{\Delta_G(u)} \cdot \frac{\PP_L(\sigma_u \sneq c)}{\PP_L(\sigma_u \seq c)} 
\ge \frac{q-2}{\Delta-1} \cdot \left[ \left( 1 - \frac{1}{q-\Delta+1} \right)^{q-\Delta+1} \right]^{\frac{\Delta-1}{q-2}} 
= \Phi(\Delta,q). 
\]
This shows the lemma. 
\end{proof}
We then give a lower bound on the key function $\Phi(\Delta,q)$ defined in \cref{lem:new-ratio-bound} when $(\Delta,q) \in \Lambda_\eps$. 
Relevant to~\cref{foot:GMP}, numerical experiments demonstrate that $\Phi(\Delta,q) < 1$ when $q = \alpha \Delta$ for $\alpha$ very close to $\alpha^*$, indicating that the current proof approach cannot go beyond $q \ge \alpha \Delta$. 
\begin{lem}\label{lem:Phi-bound}
For every $\eps > 0$ and $(\Delta,q) \in \Lambda_\eps$, we have $\Phi(\Delta,q) \ge 1+\left(1+\frac{1}{\alpha^*}\right)\eps$. 
\end{lem}

\begin{proof}
Note that the condition $q \ge \alpha \Delta + \beta$ can be rewritten as 
\begin{equation}\label{eq:q-2-bound}
q-2 \ge \alpha(\Delta-1) + \frac{\alpha}{2(\alpha^2-1)}. 
\end{equation}
First by Lemma 17 (ii) of \cite{GMP}, which can be proved directly by comparing the power series expansions, we have
\[
-(q-\Delta+1) \log\left( 1-\frac{1}{q-\Delta+1} \right) \le 1 + \frac{1}{2(q-\Delta)}. 
\]
Since we have
\[
q-\Delta = (q-2) - (\Delta-1) + 1 > (\alpha-1)(\Delta-1), 
\]
it follows that
\[
\Phi(\Delta,q) 
\ge \frac{q-2}{\Delta-1} \cdot 
\exp\left[ - \left( 1+\frac{1}{2(\alpha-1)(\Delta-1)} \right) \cdot \frac{\Delta-1}{q-2} \right]. 
\]
Notice that the right-hand side above is monotone increasing in $q-2$. Plugging in \cref{eq:q-2-bound}, we deduce that
\begin{align*}
\Phi(\Delta,q) &\ge 
\alpha \left( 1+\frac{1}{2(\alpha^2-1)(\Delta-1)} \right) \cdot 
\exp\left[ - \frac{1}{\alpha} \cdot \frac{1+\frac{1}{2(\alpha-1)(\Delta-1)}}{1+\frac{1}{2(\alpha^2-1)(\Delta-1)}} \right]\\ 
&= \alpha \left( 1+\frac{1}{2(\alpha^2-1)(\Delta-1)} \right) \cdot 
\exp\left( -\frac{1}{\alpha} -\frac{1}{2(\alpha^2-1)(\Delta-1)+1} \right)\\ 
&\ge \alpha e^{-\frac{1}{\alpha}} \cdot \left( 1+\frac{1}{2(\alpha^2-1)(\Delta-1)} \right) \cdot \left( 1-\frac{1}{2(\alpha^2-1)(\Delta-1)+1} \right)\\
&= \alpha e^{-\frac{1}{\alpha}}. 
\end{align*}
Finally, since $\alpha = (1+\eps)\alpha^*$ and $\alpha^* e^{-1/\alpha^*} = 1$, we obtain
\[
\Phi(\Delta,q) 
\ge (1+\eps) \alpha^* e^{-\frac{1}{\alpha^*} + \frac{\eps}{\alpha}} 
\ge (1+\eps) \left(1+\frac{\eps}{\alpha}\right) 
= 1+\left(1+\frac{1}{\alpha^*}\right)\eps. \qedhere
\]
\end{proof}

We are now ready to prove \cref{lem:bound-ratio}. 
\begin{lemboundratio}
\statelemboundratio
\end{lemboundratio}
\begin{proof}[Proof of \cref{lem:bound-ratio}]
The bound $\frac{1}{(1+\eps)\Delta_G(u)}$ follows from \cref{lem:new-ratio-bound} and \cref{lem:Phi-bound}. 
For the second bound, first we have the following crude bound
\[
\PP_L(\sigma_u \seq c) \le \frac{1}{|L(u)| - \Delta_G(u)} \le \frac{1}{q-\Delta}. 
\]
Therefore, 
\[
\frac{\PP_L(\sigma_u \seq c)}{\PP_L(\sigma_u \sneq c)} \le \frac{1}{q-\Delta-1}. 
\]
Since $q-2 \ge \alpha(\Delta-1)$, we deduce that
\[
\frac{q-\Delta-1}{q} \ge \frac{(q-2) - (\Delta-1)}{(q-2)+(\Delta-1)} \ge \frac{\alpha-1}{\alpha+1} \ge \frac{1}{4}.  
\]
It then follows that $\PP_L(\sigma_u \seq c) / \PP_L(\sigma_u \sneq c) \le 4/q$. 
\end{proof}

\section{Proof of \texorpdfstring{\cref{thm:high-dimensional,thm:spectraleig}}{Theorems 6 and 8}}
\label{sec:aux}
In this section, we prove \cref{thm:high-dimensional,thm:spectraleig}. We begin with the proof of the latter which is inspired by spectral arguments in \cite{Oppenheim1, ALO}. Then, in~\cref{sec:complexes} we import the relevant results for general simplicial complexes from \cite{AL,ALO}  that we will need for the proof of~\cref{thm:high-dimensional} and apply these results in \cref{sec:application} to the case of list-colorings. 

\subsection{Proof of \texorpdfstring{\cref{thm:spectraleig}}{Theorem 8}}
\label{sec:spectraleig}
In this subsection, we prove~\cref{thm:spectraleig} which we restate here for convenience.
\begin{thmspectraleig}
\statethmspectraleig
\end{thmspectraleig}
\begin{proof}
Let $n=|V|\geq 2$.From~\cref{foot:ergod}, we have that for every $(v,i)\in U_{G,L}$ we have that $\PP_{G,L} (\sigma_v \seq i)>0$. Note that for $(v,i),(w,k)\in U_{G,L}$ we have
\begin{equation*}
\wPc\big((v,i),(w,k)\big) = \begin{cases} 0,& \mbox{ if $v=w$, }\\
\frac{1}{n-1}\PP_{G,L} (\sigma_w \seq k\mid \sigma_v \seq i),& \mbox{ if $v\neq w$}, \end{cases}
\end{equation*}
since the normalizing factor for the $(v,i)$-row of $\wPc$ equals 
\[\sum_{w'\in V\backslash \{v\}}\sum_{k'\in L(w')}\PP_{G,L} (\sigma_{w'} \seq k', \sigma_v \seq i)=(n-1)\PP_{G,L} (\sigma_v \seq i).\]
Note that since $\wPc$ corresponds to the transition matrix of a random walk, for the diagonal matrix $\Db$ with diagonal entries given by $\{\PP_{G,L} (\sigma_v \seq i)\}_{(v,i)\in U_{G,L}}$ the matrix $\Ab=\Db^{1/2} \wPc \Db^{-1/2}$ is symmetric and hence an eigenvector $\zb$ with eigenvalue $\lambda$ of $\Ab$ corresponds to the right eigenvector $\Db^{-1/2}\zb$ and the left eigenvector $\Db^{1/2}\zb$ of $\wPc$. 

To study more carefully the eigenvalues and eigenvectors of $\wPc$, consider the column vectors $\ones=\{1\}_{(w,k)\in U_{G,L}}$ and $\pib=\big\{\PP_{G,L}(\sigma_w \seq k)\big\}_{(w,k)\in U_{G,L}}$, and observe that these are the right and left eigenvectors of $\wPc$, respectively, with eigenvalue 1. For $v\in V$, consider further the column vectors $\ones_v,\pib_v$ whose $(w,k)$-entries for $(w,k)\in U_{G,L}$ is equal to 1 and $\PP_{G,L} (\sigma_w \seq k)$, respectively, if $w= v$ and $k\in L(v)$, and 0 otherwise. Note that 
\[\ones=\sum_{v\in V}\ones_v, \quad \pib=\sum_{v\in V}\pib_v.\] 
For $v\in V$, observe further  that $\wPc\ones_v=\tfrac{1}{n-1}(\ones-\ones_v)$ and hence
\begin{equation}\label{eq:Pmatrixeq}
\wPc(\tfrac{1}{n}\ones-\ones_v)=\tfrac{1}{n}\ones-\tfrac{1}{n-1}\ones+\tfrac{1}{n-1}\ones_v=-\tfrac{1}{n-1}(\tfrac{1}{n}\ones-\ones_v),
\end{equation}
i.e., $\tfrac{1}{n}\ones-\ones_v$ is a right eigenvector of $\wPc$ with eigenvalue $-\frac{1}{n-1}$, from where it follows that $\tfrac{1}{n}\pib-\pib_v$ is the corresponding left  eigenvector of $\wPc$.  Let $u\in V$ be an arbitrary vertex, and 
\[S=\big\{\ones\big\}\bigcup \cup_{w\in V\backslash \{u\}}\big\{\tfrac{1}{n}\ones-\ones_w\big\}.\]
Note that $S$ consists of right eigenvectors of $\wPc$ which are linearly independent. Using the correspondence between left/right eigenvectors of $\wPc$ and eigenvectors of $\Ab$, we can extend $S$ to an eigenbasis $\overline{S}=\{\zb_t\}_{t\in U_{G,L}}$ of right eigenvectors of $\wPc$ so that eigenvectors in $\overline{S}\backslash S$ are perpendicular to the left eigenvectors corresponding to $S$, i.e., 
\begin{equation}\label{eq:eigsystem}
\mbox{ for each $\zb\in \overline{S}\backslash S$ it holds that $\pib^{\intercal}\zb=0$ and  $(\tfrac{1}{n}\pib-\pib_v)^{\intercal}\zb=0$ for $v\in V$.}
\end{equation}
Note, the equality for $v=u$ in \eqref{eq:eigsystem} follows from the fact that the vectors $\{\tfrac{1}{n}\pib-\pib_v\}_{v\in V}$ sum to the zero vector. 

The desired result will follow by showing that all right eigenvectors of $\wPc$ in $\overline{S}\backslash S$ are right eigenvectors of $\MM$ as well with the same eigenvalue multiplied by $n-1$; the right eigenvectors of $\wPc$ in $S$ are also eigenvectors of $\MM$ but correspond to the eigenvalue 0 of the latter. Recall from \cref{def:matrixM} that $\MM\big((v,i),(w,k)\big)=\PP_{G,L} (\sigma_w \seq k \mid \sigma_v \seq i) - \PP_{G,L} (\sigma_w \seq k)$ if $v\neq w$, and 0 otherwise, and hence we have that 
\begin{equation}\label{eq:4ff5fwwed3}
\MM=(n-1)\wPc-\ones\pib^{\intercal}+\sum_{v\in V}\ones_v\pib_v^{\intercal},
\end{equation}
where the subtraction of $\ones\pib^{\intercal}$ accounts for the subtraction of $\PP_{G,L} (\sigma_w \seq k)$, and the addition of $\sum_{v\in V}\ones_v\pib_v^{\intercal}$ corrects the zero terms of $\wPc$ that were affected by the subtraction. Using \eqref{eq:Pmatrixeq}, \eqref{eq:4ff5fwwed3} and the fact that $\ones$ is an eigenvector of $\wPc$ with eigenvalue 1, it is not hard to verify  that the vectors $\ones$ and $\tfrac{1}{n}\ones-\ones_v$ for $v\in V$ lie in the null space of $\MM$, and hence so do the vectors in $S$. Consider now an arbitrary right eigenvector $\zb\in \overline{S}\backslash S$ of $\wPc$ with eigenvalue $\lambda$. Then, from \eqref{eq:eigsystem}, we obtain that $\MM \zb=(n-1)\wPc \zb=(n-1)\lambda\zb$ as wanted, finishing the proof.
\end{proof}

\subsection{Proof of \texorpdfstring{\cref{thm:high-dimensional}}{Theorem 6} via high-dimensional simplicial complexes}
\label{sec:high-dimensional}
In this subsection, we prove~\cref{thm:high-dimensional}.
\subsubsection{Preliminaries on high-dimensional simplicial complexes}\label{sec:complexes}
In this section, we import results from high-dimensional complexes that we use for the proof of~\cref{thm:high-dimensional}. The presentation here follows largely~\cite{AL,ALO}.

Let $U=[n]$ be a ground set of elements. A simplicial complex $X$ is a family of subsets of $U$ which is downward-closed (under set inclusion); sets in $X$ are called faces and the dimension of a face is the set's cardinality minus one. For $k\in \mathbb{Z}_{\geq 0}$, we let $X(k)$ denote the subset of $X$ consisting of faces with dimension $k$. The simplicial complex $X$ is called pure if every maximal face (under set inclusion) has the same dimension. 

A weighted pure simplicial complex is a pair $(X,w)$ where $X$ is a pure $d$-dimensional simplicial complex $X$ and  $w: X(d)\rightarrow \mathbb{R}_{>0}$ be  a positive weight function on the maximal faces of $X$. The weight function is extended to every $\tau\in X$ by $w(\tau)=\sum_{\sigma\in X(d); \tau\subseteq \sigma} w(\sigma)$.  For a face $\tau\in X$, the link of $\tau$ is the simplicial complex $X_{\tau}=\bigcup_{\sigma\in X; \tau\subseteq \sigma}\{\sigma\backslash \tau\}$; the  maximal faces of $X_{\tau}$ inherit the weight $w_\tau$ from $w$, which is defined by $w_\tau=w(\sigma\cup \tau)$ for each $\sigma\in X_{\tau}$. The 1-skeleton of $X_\tau$ is a weighted graph with vertex set  $V_\tau=\{u\in U\mid \{u\}\in X_\tau\}$, edge set $E_{\tau}=\{\{u,u'\}\mid u\neq u' \mbox{ and }\{u,u'\}\in X_{\tau} \}$, and weights on the edges given by  $w_{\tau}(\{u,u'\})$ for $(u,u')\in  X_{\tau}$.

For a weighted pure  simplicial complex $(X,w)$ with dimension $d$, we can define Glauber dynamics $(\sigma_t)_{t\geq 0}$ on the maximal faces of $X$ as follows. Start from an arbitrary maximal face $\sigma_0\in X(d)$. At each time $t\geq 0$ update the current face $\sigma_{t}\in X(d)$ to $\sigma_{t+1}\in X(d)$ by selecting an element $i\in \sigma$ uniformly at random, and setting $\sigma_{t+1}=\sigma_t\cup \{i\}\backslash \{j\}$, where $i\in U$ is chosen with probability proportional to $w(\sigma_t\cup \{i\}\backslash \{j\})$. We let $\Pc^{X,w}$ denote the transition matrix of the Glauber dynamics on $(X,w)$. 

For $\alpha\in [0,1)$, a face $\tau$ of $X$ is an $\alpha$-spectral expander of $(X,w)$ if the second largest eigenvalue of the simple non-lazy random walk on the 1-skeleton of $X_\tau$ is at most $\alpha$. The main theorem we will use about high-dimensional complexes is the following.
\begin{thm}[{\cite[Theorem~1.5]{AL}}]\label{thm:simplicialcomplex}
Let $(X,w)$ be a weighted pure $d$-dimensional simplicial complex and let $\Pc=\mathcal{P}^{X,w}$ be the transition matrix of Glauber dynamics on $(X,w)$. Suppose that for $k=-1,0,\hdots,d-2$ there exists $\alpha_{k}\in [0,1)$ such that  every $k$-dimensional face $\tau\in X(k)$ is an $\alpha_k$-spectral expander of $(X,w)$. 

Then, $\lambda_2(\Pc)\leq 1-\frac{1}{d+1}\prod^{d-2}_{k=-1}(1-\alpha_k)$.
\end{thm}

\subsubsection{Application to list-colorings and proof of \texorpdfstring{\cref{thm:high-dimensional}}{Theorem 6}}
\label{sec:application}
We are now ready to prove~\cref{thm:high-dimensional} which we restate here for convenience.
\begin{thmhighdimensional}
\statethmhighdimensional
\end{thmhighdimensional}
\begin{proof}
A list-coloring  instance $(G,L)$ with $G=(V,E)$ and $L=\{L(v)\}_{v\in V}$ can be viewed as a weighted pure simplicial complex $(X_{G,L},w_{G,L})$ as follows. 

The ground set of elements is going to be the set $U_{G,L}=\{(v,i)\mid v\in V, i\in L(v)\}$. Then, a subset $S=\{(v_{1},i_1),\hdots, (v_{s},i_s)\}$ of $U_{G,L}$ is in 1-1 correspondence with a (partial) coloring assignment where vertex $v_j$ gets the color $i_j$ for $j=1,\hdots, s$. 

We let $X=X_{G,L}$ be the subsets of $U$ which correspond to the set of all partial list-colorings of the instance $(G,L)$, cf. \cref{foot:partial}. Then $X$ is a downward-closed collection of subsets; in fact, $X$ is a pure $(n-1)$-dimensional complex, since every element of $X$ contains at most 1 element from each of the sets $S_v=\{(v,i)\mid v\in V, i\in L(v)\}$ for $v\in V$  and maximal faces contain exactly one (since they correspond to the set of list-colorings $\Omega_{G,L}$). We let $w_{G,L}$ equal 1 for all maximal faces of $X$.

With these definitions, it remains to note that transitions for Glauber dynamics on $(G,L)$ are in 1-1 correspondence with transitions for Glauber dynamics on $(X_{G,L},w_{G,L})$, and that local expansion of $\alpha_s$ for a conditioned list-coloring instance $(G_\tau,L_{\tau})$ for some partial list-coloring $\tau$ on $S\subseteq V$ with $|S|=s$ translates into the face $\tau\in X_{G,L}(s-1)$ being an $\alpha_s$-spectral expander. Hence, the result follows by applying~\cref{thm:simplicialcomplex}.
\end{proof}

\end{document}